\documentclass[
  aps,
  prx,
  onecolumn,
  superscriptaddress,
  nofootinbib,
  longbibliography
]{revtex4-2}

\usepackage[T1]{fontenc}
\usepackage[utf8]{inputenc}
\usepackage{amsmath,amssymb,amsthm,mathtools,bm}
\usepackage{graphicx}
\usepackage{xcolor}
\usepackage{microtype}
\usepackage{enumerate}
\usepackage[colorlinks=true,linkcolor=blue,citecolor=blue,urlcolor=blue]{hyperref}

\newtheorem{theorem}{Theorem}
\newtheorem{lemma}{Lemma}
\newtheorem{proposition}{Proposition}
\newtheorem{corollary}{Corollary}
\newtheorem{definition}{Definition}

\newcommand{\CRAS}{\mathsf{CRAS}}
\newcommand{\BQNCf}{\mathsf{B\text{-}QNC}_{f}^{0}}
\newcommand{\BPP}{\mathsf{BPP}}
\newcommand{\poly}{\mathrm{poly}}
\newcommand{\Z}{\mathbb{Z}}

\newcommand{\ket}[1]{\lvert #1\rangle}

\newcommand{\abs}[1]{\left\lvert #1\right\rvert}

\newcommand{\ind}[1]{\mathbf{1}[#1]}
\newcommand{\xor}{\oplus}
\newcommand{\AND}{\operatorname{AND}}
\newcommand{\OR}{\operatorname{OR}}
\newcommand{\Thr}{\operatorname{TH}}
\newcommand{\WTH}{\operatorname{WTH}}

\newcommand{\PKU}{Center on Frontiers of Computing Studies, School of Computer Science, Peking University, Beijing 100871, China}
\newcommand{\bnu}{School of Artificial Intelligence,
 Beijing Normal University, Beijing,
 100875, China}

\newcommand{\scnu}{School of Physics, South China Normal University, Guangzhou 510006, China}

\begin{document}

\title{Quantum Advantage with Adaptive Shallow Circuits}
%\title{On the Complexity of Quantum Adaptive Shallow Circuits}
\author{Yusen Wu}
 \email{yusen.wu@bnu.edu.cn}
 \affiliation{\bnu}

 \author{Yukun Zhang}
 \email{yukunzhang@stu.pku.edu.cn}
 \affiliation{\PKU}

 \author{Xiaoming Zhang}
 \email{xiaomingzhang@scnu.edu.cn}
 \affiliation{\scnu}

 \author{Chuan Wang}
 \email{wangchuan@bnu.edu.cn}
 \affiliation{\bnu}

 \author{Xiao Yuan}
 \email{xiaoyuan@pku.edu.cn}
 \affiliation{\PKU}

%\date{\today}

\begin{abstract}

Quantum advantage is widely expected to require sufficiently deep circuits, where correlations and global computational structure can grow beyond the reach of efficient classical simulation. This expectation is especially stark for constant-depth circuits with local readout: the expectation value of any fixed local observable lies within a bounded backward lightcone and is therefore classically tractable. Here we show that measurement feedback changes this picture. We establish a strict hierarchy of computational power: at fixed coherent depth, increasing the number of feedback outcomes strictly enlarges the class of functions accessible through a local expectation value. The two ends of this hierarchy exhibit distinct computational regimes. With logarithmic feedback, local expectation values for product-state inputs are efficiently classically simulable. Polynomial feedback, by contrast, enables an explicit family of adaptive shallow circuits to encode prime-field discrete logarithm problem~(DLP) into a fixed single-qubit expectation. Assuming the standard worst-case classical hardness of DLP, estimating this expectation value is classically hard. These results reveal a feedback-driven complexity transition, with further implications for resource lower bounds on DLP and the complexity of local-observable estimation under area-law entanglement. Our results open a new route to quantum advantage with shallow quantum circuits.

\end{abstract}

\maketitle

\section{Introduction}
Shallow quantum circuits provide a natural setting in which to ask where quantum
computational advantage begins. Constant-depth quantum circuits can already
outperform constant-depth classical circuits on explicit relation
problems~\cite{Bravyi2018}, with extensions to average-case
separations~\cite{LeGall2019} and noisy geometrically local
architectures~\cite{Bravyi2020}. More broadly, most proposals for quantum
advantage in shallow or near-term architectures are certified through sample from the output state---random circuit
sampling~\cite{Boixo2018,Arute2019Sycamore,Wu2021Zuchongzhi}, Gaussian boson
sampling~\cite{hamilton2017gaussian,deshpande2022quantum,liu2026gaussian,madsen2022quantum}
and instantaneous quantum polynomial-time~(IQP)
circuits~\cite{BremnerJozsaShepherd2011,BremnerMontanaroShepherd2016,BremnerMontanaroShepherd2017}
all rely on the difficulty of generating samples from
a high-dimensional output
distribution~\cite{AaronsonChen2017Supremacy,Bouland2019,Movassagh2023,HangleiterEisert2023}.
Related ideas are extended to thermal states, where constant-temperature Gibbs
sampling can encode a robust IQP-type sampling
problem~\cite{BergamaschiChenLiu2024}. Across these settings, quantum advantage is generally achieved through a global measurement.

However, many practical problems involve local measurements rather than access to global samples, leading to a fundamentally different paradigm.
For a non-adaptive
constant-depth circuit built from local quantum gates and a product  input state, the
expectation value of a local observable is confined to a bounded backward light
cone and is therefore efficiently computable classically. 
%This raises a sharper question: \emph{Can quantum advantage still hold if the output is given by local measurement?} 
Even beyond constant
depth, local observables can be much easier to estimate than full output
distributions. For example, local observables in typical random quantum
circuits can be approximated classically in regimes where sampling from the
full distribution is believed to remain
hard~\cite{angrisani2025classically}. Hardness of quantum random sampling tasks
therefore does not by itself imply hardness of local-observable estimation.
This raises a sharper question: \emph{What's the minimal requirement to achieve quantum advantage with only local measurements?} 

The picture may change once mid-circuit measurements and classical feedforward
are allowed. These operations are becoming standard ingredients of modern
quantum platforms, particularly in quantum error correction and
measurement-based protocols~\cite{Acharya2025,Putterman2025,Bluvstein2026}.
Computations in this setting naturally consist of short coherent layers
interleaved with measurements and classical processing. Unlike its
non-adaptive counterpart, an adaptive shallow circuit is not constrained by a
single static backward light cone. Measurement outcomes generated in spatially
separated regions can be collected and processed by a classical controller,
and then used to determine later local quantum operations. Related
LOCC-assisted and dynamic-circuit constructions show that measurements and
feedforward can generate long-range correlations or implement long-range
entangling operations within shallow quantum
depth~\cite{piroli2021quantum,piroli2024approximating,Yan2025,BaeumerPRXQ2024,BaeumerPRL2024,BaeumerWoerner2025}.
Recent work has further shown that these resources can preserve classically
hardness in random sample problem in shallow depth
architectures~\cite{cao2026measurement}. Whether they can also make a local
expectation value classically hard has remained open.

In this work, we show that the amount of classical feedforward induces a
genuine computational transition in adaptive shallow circuits with local
readout. We study constant-round adaptive shallow circuits ($\CRAS$), in which constant-depth coherent layers are interleaved with
measurements and polynomial-time classical feedforward, and the final output
is the expectation value of a fixed single-qubit observable. Our first result
establishes a strict computational hierarchy: increasing the number of
intermediate measurement outcomes used for feedforward strictly enlarges the
local-readout power of adaptive shallow circuits. We identify two distinct regimes of classical complexity. With only $\mathcal O(\log n)$ feedforward outcomes, the local expectation value remains efficiently computable classically. By contrast, assuming the worst-case classical hardness of the prime-field discrete logarithm problem~(DLP), we construct an explicit family of circuits in $\CRAS$ for which no randomized classical polynomial-time algorithm can estimate the expectation value of one fixed single-qubit observable to constant additive error. Thus, increasing classical feedforward can drive local readout from a classically tractable regime to one that is classically hard. Finally, we show that solving an $n$-bit DLP in this setting requires $\Omega(n)$ intermediate measurement outcomes to be fed forward, establishing that extensive classical feedback is necessary for this task. We summarize our main results in Fig.~\ref{fig:overview}.

We further show that the hard family admits a two-dimensional
nearest-neighbour realization: all coherent quantum gates are geometrically
local, and the only nonlocal resource is global classical feedforward. This
geometric refinement connects our result to shallow circuits assisted by local operations and classical communication~(LOCC), whose post-measurement
trajectories obey entanglement area laws~\cite{piroli2021quantum}. Surprisingly, the quantum state associated with each measurement trajectory obeys an entanglement area law, while the single-qubit expectation value obtained by averaging over all trajectories remains classically hard to estimate in the worst case.
%Surprisingly, every measurement trajectory in our construction obeys an entanglement area law and is therefore efficiently simulable classically, while the average expectation value remains, in general, classically hard to predict.
This mechanism differs from the standard
tensor-network route to area-law hardness, where the difficulty arises from
contracting general two-dimensional
PEPS~\cite{schuch2007computational,haferkamp2020contracting,harley2025computational}. Here, hardness arises not from a highly entangled static state, but from a global adaptive process that concentrates distributed computational information into one local observable. Our results therefore answer, in the noiseless worst-case setting, the
question raised by Mele \emph{et al.}~\cite{mele2026noise} of how the
classical-simulation picture changes once mid-circuit measurements and
classical feedforward are allowed.

\section{Main Results}
\begin{figure}[t]
    \centering
    \includegraphics[width=\linewidth]{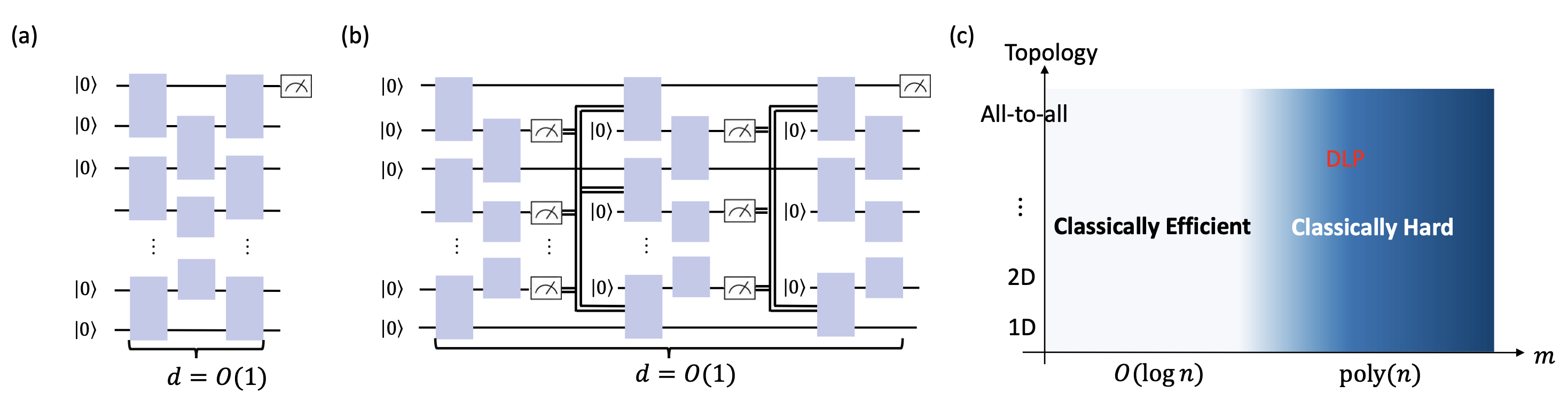}
    \caption{(a)~In a non-adaptive constant-depth circuit, the expectation value of a fixed
local observable is determined by a bounded backward light cone and is
therefore classically computable. 
(b)~Mid-circuit measurements and feedforward change the causal structure:
measurement outcomes generated across the system are collected by a classical
controller, processed globally, and used to determine later operations. The final output is determined by a local observable. (c)~The figure provides a high-level overview of our main results. Theorem~\ref{thm:hierarchy} establishes a strict hierarchy of computational power as the amount of classical feedback $m$ increases. Theorem~\ref{thm:main-maintext} establishes worst-case classical hardness, whereas Theorem~\ref{lowerbound} provides a resource lower bound for solving DLP within this model.}
    \label{fig:overview}
\end{figure}

In the context of $\CRAS$, the amount of classical feedback emerges as a natural
computational resource. The underlying intuition is causal: if an adaptive
circuit has total two-qubit circuit depth $d$ and uses $m$ intermediate
measurement outcomes as classical feedback, then the final local readout is supported on at most
$\mathcal{O}((m+1)2^d)$ input qubits. This raises a natural question: does
increasing the amount of classical feedback strictly enhance the computational
power of adaptive shallow circuits? The following theorem answers this
question affirmatively.

\begin{theorem}[Feedback hierarchy]
\label{thm:hierarchy}
Fix any constant $d\geq 1$, and let $\CRAS_d[m]$ denote the class of
$\CRAS$ circuit with $d$ circuit depth and using at most $m$ intermediate
measurement outcomes as classical feedback. Then, for every $m\geq 2$,
\begin{align}
    \CRAS_d[o(m)]
    \subsetneq
    \CRAS_d[m].
\end{align}
%Note that the standard constant-depth circuit ${\mathsf{B\text{-}QNC}^{0}}=\CRAS_d[0].
In other words, there exists a family of functions that can be computed by $\CRAS$ circuits with $m$ bits of classical feedback, but cannot be computed by any $\CRAS$ circuit without, or with $o(m)$ bits of classical feedback. 
\end{theorem}
%上次和xiaoming讨论说写成hierachy 1\subset 2...\subset m这种，但我翻了一下证明，现在这么写可能更严格一些，我们证明的是feedback数量是Omega m，所以这个常数不一定是差了1.
The hierarchy concerns quantum-state inputs whose information is not directly available to the classical controller. If the input is instead supplied to the controller as a classical bit string, it can be processed and written to the output without intermediate quantum measurements, and classical feedback no longer imposes the same restriction. The light-cone bound implies that $\CRAS$ with $m=\mathcal{O}(\log n)$ classical feedback is efficiently classically simulable in $\mathcal{O}(n^{2^d})$ running time. This naturally raises the question of whether its local quantum mean value can become classically hard as $m$ increases.

Hardness results for random circuit sampling, IQP circuits and related models typically concern sampling from the full output distribution or estimating individual output probabilities~\cite{Bouland2019,BoulandFeffermanLandauLiu2022FOCS}. Estimating a fixed local observable is a different task, since its expectation value depends only on the reduced state on the support of the observable. Sampling hardness therefore does not directly imply hardness of estimating a local expectation value to either constant or inverse-polynomial additive accuracy. A similar distinction arises in the one-clean-qubit model (${\rm DQC}_1$)~\cite{morimae2014hardness}, where the readout of the clean qubit encodes the normalized trace $2^{-n}{\rm Tr}[U]$. Hardness results based on multiplicative approximation of related quantities~\cite{takahashi2013hardness,morimae2014hardness} do not imply additive hardness for the normalized trace: when the latter is exponentially small, even the estimate zero achieves inverse-polynomial additive accuracy. These examples illustrate the basic obstruction: hardness of a global measurement does not automatically carry over to local-observable estimation.

Here, we use the prime-field discrete logarithm problem (DLP) as the basis of our hardness result. Let $\mathcal I:=\{(p,g,y):p\text{ is prime},\,g\text{ generates }\mathbb Z_p^\ast,\,y\in\mathbb Z_p^\ast\}$ denote the set of valid instances, where $\mathbb Z_p^\ast=\{1,2,\ldots,p-1\}$ is the multiplicative group of nonzero residues modulo $p$. For $(p,g,y)\in\mathcal I$, let $r=p-1$ and $n=\lceil\log_2 p\rceil$. Since $g$ generates $\mathbb Z_p^\ast$, there is a unique $x\in\mathbb Z_r=\{0,1,\ldots,r-1\}$ satisfying $g^x\equiv y\pmod p$. The DLP asks to recover $x=\log_g y$ from $(p,g,y)$. Without loss of generality, we assume that DLP is not in $\mathsf{BPP}$. The conjectured classical hardness of the DLP has been widely used to construct quantum algorithms and establish their advantage over classical computation~\cite{Shor1994,Shor1997,Liu2021}.

\begin{theorem}[Classical hardness]
\label{thm:main-maintext}
There exists a quantum circuit in $\CRAS$, whose local output expectation values cannot be
efficiently estimated classically within constant additive error.
\end{theorem}

Several features of the theorem are worth emphasizing. First, the involved quantum gates of $C_{p,g,y}$ is explicit from the classical instance $(p,g,y)$. There is no oracle access and no
hidden distribution over inputs compared with previous results~\cite{HoyerSpalek2005,Rosenbaum2013}. Second, the final output is local, where the
same designated output qubit and the same Pauli-$Z$ observable are used for
all valid instances. From a high-level perspective, the construction proceeds as follows. A constant-depth adaptive quantum subroutine prepares two quantum states $|I^{(k)}_x\rangle=2^{-k/2}\sum_{i=0}^{2^k-1}|x+i~{\rm mod}r\rangle$ and $|I^{(k)}_{x_0}\rangle=2^{-k/2}\sum_{i=0}^{2^k-1}|x_0+i~{\rm mod}r\rangle$, where $k=n-3$ and $x_0=(p+1)/2$. Then the modified swap-test is used to estimate their overlap $|\langle I_x^{(k)}|I_{x_0}^{(k)}\rangle|^2$ that can determines whether $x\in [1,(p-1)/16]$ or $x\in[(p-1)/2+1,(p-1)/2+(p-1)/16]$. By the standard reduction from these promise regions to full prime-field discrete logarithm, any constant-accuracy classical estimator for the local observable would solve the original discrete-logarithm problem. We leave technical details to Appendix~\ref{proofsketchtheorem2}. Known hardness results for sampling shallow-circuit outputs do not imply Theorem~\ref{thm:main-maintext}, because such hardness need not survive compression to a fixed local expectation value. Specifically, sampling the full output of certain shallow quantum circuits may be classically hard~\cite{BermejoVegaEtAl2018},
whereas any efficiently computable Fourier-sparse Boolean function of the same outputs is classically tractable by combining Refs.~\cite{BravyiGossetMovassagh2021,TakahashiEtAl2026}.

The same construction also applies to modular order finding, which asks for the smallest positive integer $r$ such that $a^r \equiv 1 \pmod N$ for a given $a \in \mathbb{Z}_N^{*}$. Integer factorization, a problem of central importance in computational number theory and modern cryptography, reduces to modular order finding in randomized polynomial time. Its modular-exponentiation step follows a similiar construction used for DLP, while the remaining Fourier sampling and classical postprocessing are compatible with $\mathrm{CRAS}$. The construction above uses DLP as the underlying assumption to the classical hardness. This motivates a related question: what quantum resources does DLP require in the context of adaptive shallow quantum circuit? Existing lower bounds for DLP primarily constrain the number or depth of group operations in generic models~\cite{Shoup1997,HhanYamakawaYun2024}. The following theorem reveals a different limitation of shallow adaptive quantum computation: constant depth and a constant number of rounds require ancillary space linear in the input size.

\begin{theorem}[Feedback lower bound]
\label{lowerbound}
Let the prime $p>2$, and $n=\lceil\log_2 p\rceil$.
Any input-oblivious $\CRAS$ protocol with constant round and circuit depth that solves DLP with success probability
at least $2/3$ requires $\Omega(n)$ classical feedback.
\end{theorem}

Here, the input-oblivious means that $y$ enters only through $|y\rangle$, rather than as a classical input. Taken together, our results identify measurement feedback as a control parameter for adaptive shallow quantum circuits with local readout. At fixed coherent depth, we establish a strict hierarchy of computational power: increasing the feedback budget enlarges the family of functions that can be encoded in a local observable. The two ends of this hierarchy exhibit qualitatively different computational regimes. With $m=\mathcal O(\log n)$ feedback outcomes, local expectation values for product-state inputs remain classically tractable. With $m=\operatorname{poly}(n)$, by contrast, the model can encode a DLP-hard task into a single local expectation value. Assuming the classical hardness of DLP, measurement feedback therefore drives a transition from classically tractable to classically intractable local readout. Finally, our lower bound for input-oblivious protocols provides a complementary perspective on DLP: even with adaptive quantum computation, solving the problem at constant depth requires feedback and ancillary space that scale linearly with the input size.

\section{Implications}

Our results have further implications for the roles of geometry and entanglement in adaptive shallow quantum computation. Theorem~\ref{thm:main-maintext} is formulated for the geometry-agnostic $\CRAS$ model, in which two-qubit gates may act between arbitrary pairs of qubits. Nevertheless, the hardness persists when coherent interactions are restricted to a two-dimensional nearest-neighbour architecture. The arithmetic component of our construction is expressed in terms of clean reversible updates, including fan-out and threshold-type operations. Threshold-type operations admit constant-depth two-dimensional realizations in the unbounded-fan-out model, while fan-out and long-range CNOT primitives can themselves be implemented at constant depth using mid-circuit measurements and classical feedforward. Further details are provided in Appendix~\ref{app:2d-crossbar} and Refs.~\cite{HoyerSpalek2005,BaeumerWoerner2025}.

\begin{corollary}[Classical hardness with 2D structure]
\label{coro:2Darchitecture}
There exists an explicit family of $\CRAS$
implemented on a 2D architecture, whose output expectation values cannot be
efficiently estimated classically within constant additive error.
\end{corollary}

This geometric refinement should be distinguished from the previously known
observation that Shor-type computations can be compressed into shallow
geometrically local architectures~\cite{Rosenbaum2013}. Our purpose here is not merely to implement
the discrete-logarithm algorithm in shallow depth, and we provide an explicit quantum circuit construction, without oracle access. This geometric formulation naturally connects our result to area-law entanglement in
LOCC-assisted shallow circuits. Ref.~\cite{piroli2021quantum} showed that
fixed-depth geometrically local quantum circuits assisted by local operations
and classical communication obey entanglement area laws. Our two-dimensional
implementation has the same architecture. Consequently, conditioned on any
individual quantum trajectory, the post-measurement state belongs to the
area-law class of adaptive shallow circuit. Nevertheless, the
unconditional local readout generated by the full adaptive process remains
classically hard to predict in the worst case. We note that this is distinct from the familiar tensor-network route to area-law hardness. General two-dimensional PEPS obey entanglement area laws, yet their exact contraction and exact evaluation of local observables are $\#P$-complete, with worst-to-average reductions for random ensembles~\cite{schuch2007computational,haferkamp2020contracting}. More recently, constant-precision estimation of local observables for sufficiently weakly injective PEPS was shown to be ${\rm postBQP}$-complete~\cite{harley2025computational}. Here, by contrast, the hardness appears in the unconditional local response of an explicit adaptive shallow process, even though every conditioned trajectory remains within the area-law regime.

\begin{corollary}[Entanglement area law hardness]
\label{coro:area-law}
There exists an explicit family of 2D quantum states obey the entanglement area law, while the expectation value of one fixed single-qubit observable is
classically hard in the worst case.
\end{corollary}

\section{Discussion}
Our results characterize how the computational power of adaptive shallow quantum circuits with local readout changes with the amount of classical feedback. At fixed coherent depth, increasing the feedback budget strictly enlarges the class of functions that can be accessed through a local readout. This hierarchy contains two qualitatively different regimes. With $m=\mathcal O(\log n)$ feedback, local expectation values for product-state inputs can be estimated efficiently by a classical algorithm. With $m=\operatorname{poly}(n)$, adaptive shallow circuits can encode DLP into a single-qubit expectation value, making its estimation classically hard under the classical hardness assumption for DLP. Classical feedback can therefore change local-readout estimation from a classically tractable problem to a classically intractable one without increasing the coherent depth.

We further use adaptive shallow circuits to establish a resource lower bound for DLP. Any constant-depth protocol for the DLP task considered here requires a feedback budget $\Omega(n)$, corresponding to $\Omega(n)$ auxiliary qubits. Rigorous lower bounds for DLP are known mainly in restricted computational settings~\cite{Shoup1997,HhanYamakawaYun2024}; our result identifies a different obstruction, namely the feedback required at constant quantum depth. The scale of this bound is particularly revealing: classical simulation is efficient with $m=\mathcal O(\log n)$ feedback, whereas DLP requires $m=\Omega(n)$. DLP therefore lies far beyond the classically tractable regime in the context of adaptive shallow quantum circuit. Our construction also provides a circuit-based route to worst-case hardness of local-observable estimation under area-law entanglement. It is proved that quantum states generated by the two-dimensional adaptive circuit obey the entanglement area law, yet the expectation value produced by the full adaptive process remains classically hard to estimate. This differs from the standard PEPS-based route: general two-dimensional PEPS can encode postselected quantum computations, with exact contraction being $\#P$-complete and local-observable estimation becoming PostBQP-complete in suitable promise settings~\cite{schuch2007computational,harley2025computational}. Here, the hardness arises within the more structured setting of an explicit adaptive shallow circuit with a fixed local readout.

Our results separate a classically simulable regime with logarithmic feedback from a classically hard regime with polynomial feedback, leaving the intermediate range unresolved. A natural question is how this boundary changes when the classical controller is restricted to $\mathrm{NC}^0$ or $\mathrm{TC}^0$ computation. Meanwhile, recent work on one-dimensional MERA-like quantum circuits provides efficient classical simulations when the local gates are random~\cite{bermejo2026quantum}, but whether similar guarantees hold in the worst case remains open. Another important question is robustness to noise. Constant-strength gate and measurement errors may weaken the signal generated by feedback and thus restore classical simulability. It is therefore important to determine whether the hardness persists below a finite noise threshold and whether a noise-robust average-case separation can be established. Resolving these questions would bring the present results closer to experimentally relevant adaptive quantum circuits.

\section*{Data availability}
\noindent 
No datasets were generated or analysed during the current study.

\section*{Competing interests}
\noindent The authors declare no competing interests

\section*{Acknowledgments}
\noindent Y.~Wu acknowledges the support from the NSFC Grant (No.~62501060).  X.~M.~Zhang acknowledges the support from the NSFC Grant (No.~12405013) and the Guangdong Provincial Quantum Science Strategic Initiative (Grants No.GDZX2503008, No.GDZX2503001).  C.~Wang acknowledges the support from the NSFC Grant (Nos.~62461160263 and~62371050). X.~Yuan acknowledges the support from the National Natural Science Foundation of China Grant (No.~12361161602) and NSAF Grant (No.~U2330201), the Quantum Science and Technology-National Science and Technology Major Project (2023ZD0300200), Beijing Natural Science Foundation Z250004, and 
Beijing Science and Technology Planning Project (Grant No.~Z25110100810000).

\clearpage
\bibliography{main}
\clearpage
\appendix

\section{Comparison with Related Works}

Our result should be distinct from recent measurement-based quantum-advantage
proposals. Measurement-driven fan-out staircases can compress dense IQP-type
sampling circuits to constant depth~\cite{cao2026measurement}, showing that
measurements and feedforward can overcome light-cone restrictions in shallow
architectures. However, the advantage in those constructions is still
certified through random sampling. Other works have studied oracle
separations or the effect of intermediate measurements on restricted circuit
models~\cite{chia2023need,jozsa2024iqp}. By contrast, our
focus is neither sampling or computational separation in the context of oracle, we instead give an
explicit construction showing that adaptive measurement and classical
feedforward can compress classically intractable information
into the mean value of a fixed local observable. Conceptually, adaptively measurement identifies a sharp complexity transition of the local readout for constant-depth quantum circuit, paving a novel way for achieving computational quantum advantage.

Relative to Shor's algorithm, our theorem preserves the hidden-structure origin of the hardness while altering both the computational architecture and the form of the output. Shor's original algorithm is most naturally formulated as a coherent quantum circuit followed by classical number-theoretic postprocessing~\cite{Shor1994,Shor1997}. By contrast, in our setting the coherent quantum component is restricted to constant depth, while a polynomial-time classical controller orchestrates the adaptive measurements and feedforward. Moreover, the output is no longer an explicitly recovered period or exponent register; instead, the relevant information is compressed into the expectation value of a fixed single-qubit observable. From this perspective, the theorem suggests that hidden-subgroup-type hardness can persist even under a remarkably severe restriction on coherent depth, provided that intermediate measurement and classical feedforward are available. Ref.~\cite{Rosenbaum2013} proposed a method for compiling Shor's algorithm into a constant-depth quantum circuit with classical feedforward, but did not provide an explicit gate-level construction for the arithmetic components involved in our setting. In comparison, our result mainly focuses on the quantum advantage arising from local-observable estimation, and also provides an explicit construction of the required quantum gates, such as the multiplier and modular-reduction functions. We refer readers to Appendix~\ref{proofsketchtheorem2}.

Random quantum circuit sampling emerged as a leading benchmark for near-term quantum advantage because random circuits are native to programmable platforms and scale naturally in experiment, as reflected in the series of quantum advantage experiments~\cite{Boixo2018,Arute2019Sycamore,Wu2021Zuchongzhi,Morvan2024PhaseTransitionsRCS}. Its complexity-theoretic justification is correspondingly more demanding: the task is to approximately sample from an exponentially large output distribution, and the associated hardness result relies on the worst-case hardness assumption followed by worst-to-average case reduction~\cite{AaronsonChen2017Supremacy,Bouland2019,DalzellHunterJonesBrandao2022PRXQ,Movassagh2023HardnessRQC,KondoMoriMovassagh2022FOCS,BoulandFeffermanLandauLiu2022FOCS,AharonovGaoLandauLiuVazirani2023STOC,HangleiterEisert2023RMP}. By contrast, our task is intentionally much narrower. We do not ask for approximate sampling of the full distribution, but only for the expectation value of a fixed local Pauli observable. More importantly, in our construction a comparatively coarse additive error already suffices to witness hardness, rather than a high-fidelity sampling guarantee, generally characterized by the cross-entropy benchmarking~\cite{Boixo2018NearTermSupremacy,Arute2019Sycamore,Villalonga2019qFlex,PanChenZhang2022Sycamore,GaoKalinowskiChouLukinBarakChoi2024XEB,ZlokapaVillalongaBoixoLidar2023Boundaries}. The trade-off is that our theorem is formulated as a worst-case reduction from discrete logarithm rather than the hardness of total-variation distance statement. The conceptual payoff is that, once intermediate measurement and polynomial-time classical feedforward are allowed, nontrivial quantum hardness can already survive in a single local deterministic readout, without requiring access to an exponentially large output distribution.

Finally, our theorem is complementary to shallow-circuit separations in the style of Bravyi \emph{et al.}\cite{Bravyi2020}. Those works prove a striking unconditional separation for noisy shallow quantum circuits, but against restricted shallow classical models and for relation problems tailored to depth-complexity comparisons. Our theorem goes in the opposite direction: the classical comparison class is unrestricted randomized polynomial time, but the hardness statement is conditional and the quantum family is explicit rather than generic. The two viewpoints illuminate different facets of shallow quantum advantage. One shows that even noisy shallow circuits can outrun low-depth classical computation; the other shows that adaptive shallow circuits can compress classically inaccessible information into a fixed local observable.

\section{Problem Statement and Quantum Circuit model}

\subsection{Prime-field discrete logarithm}

\begin{definition}[Prime-field DLP]
Let $\mathcal I := \{(p,g,y): p \text{ prime},\ g \text{ generates } \Z_p^\ast,\ y\in\Z_p^\ast\}$, where the multiplicative group of nonzero residues $\Z_p^\ast=\{1,2,\cdots,p-1\}$. For $(p,g,y)\in\mathcal I$, let $r=p-1$ and $n=\lceil \log_2 p\rceil$. Since $\Z_p^\ast$ is cyclic of order $r$, there exists a unique residue class $x\in\Z_r=\{0,1,\cdots,r-1\}$ such that $g^x \equiv y \pmod p$. The prime-field DLP asks to compute $x=\log_g y\pmod p$ from the input $(p,g,y)$.
\end{definition}

\begin{definition}[Promise problem $\mathrm{DLP}_\delta(p,g)$]
Fix $\delta\in(0,1/2]$. For a prime $p$ and a generator $g\in\Z_p^\ast$, the problem $\mathrm{DLP}_\delta(p,g)$ has inputs $y=g^x\bmod p$ under the promise that
\begin{align}
    x\in\{1,\ldots,\lfloor \delta(p-1)\rfloor\}
\end{align}
or
\begin{align}
    x\in\left\{\frac{p-1}{2}+1,\ldots,\frac{p-1}{2}+\lfloor \delta(p-1)\rfloor\right\}.
\end{align}
The $\mathrm{DLP}_\delta(p,g)$ task is to distinguish these two cases.
\end{definition}

Actually, these two problems are polynomially equivalent to each other.
\begin{lemma}[\cite{blum2019generate}]
    For every $p$, generator $g$ and  $\delta\in[1/{\rm poly}(n),1/2]$, if there exists a polynomial time algorithm for ${\rm DLP}_{\delta}(p,g)$, then there exists a polynomial time algorithm for ${\rm DLP}(p,g)$.
    \label{lemma:DLPReduction}
\end{lemma}

\subsection{Constant-round adaptive shallow circuits}

\begin{definition}[$\CRAS$ class]
A polynomial time-uniform family of adaptive quantum circuits $\{C_z\}_{z\in\{0,1\}^\ast}$ belongs to $\CRAS$ if, on inputs $z$ of length $n$, the following conditions hold.
\begin{enumerate}[(i)]
    \item $C_z$ performs on at most $\poly(n)$ qubits.
    \item The circuit is partitioned into $\mathcal{O}(1)$ adaptive rounds.
    \item In each round, the quantum circuit has depth $\mathcal{O}(1)$ over arbitrary long-range single- and double-qubit gates.
    \item Intermediate computational-basis measurements are allowed.
    \item Between rounds, measurement outcomes may be processed by an arbitrary polynomial-time classical controller, which can determine later quantum operations.
    \item At the end, one designated output qubit $q_{\rm out}$ is measured in the Pauli $Z$ basis and the quantity of interest is $\mu_z= \langle Z_{q_{\rm out}}\rangle_{C_z}$.
\end{enumerate}
The family is \emph{explicit} if a polynomial-time classical algorithm outputs the circuit description of $C_z$ from $z$.
\end{definition}

\section{Preliminary Knowledge}

The following results support the proof of our main result.

\begin{lemma}[H{\o}yer and {\v S}palek~\cite{{HoyerSpalek2005}}]
\label{fact:HS}
For polytime-uniform constant-depth polynomial-size quantum circuits with unbounded fan-out, the following primitives are available in $\BQNCf$ (bounded error constant-depth quantum circuit with fan-out gate):
\begin{enumerate}[(i)]
    \item addition, multiplication of many integers and division with remainder~\cite[Theorem~4.8]{HoyerSpalek2005},
    \item reversible modular addition~\cite[Lemma~4.10]{HoyerSpalek2005},
    \item the quantum Fourier transform ${\rm QFT}_q$ for arbitrary modulus $q$~\cite[Theorem~4.17]{HoyerSpalek2005}.
\end{enumerate}
These subroutines are coherent in the sense required here: the input registers are preserved, the answers are written to designated ancillas, and the subroutines are used reversibly inside the proof of arbitrary-modulus Fourier transforms.
\end{lemma}

\begin{lemma}[Dynamic constant-depth fan-out and CNOT ladders~\cite{BaeumerWoerner2025}]
\label{fact:BW}
Unbounded fan-out, ${\rm CNOT}$ ladders and long-range CNOT gates can be implemented by constant-depth dynamic circuits using only one- and two-qubit gates, parallel mid-circuit measurements and one round of classical feedforward.
\end{lemma}

\begin{lemma}
\label{lem:compile}
Let $\{U(n)\}$ be a polynomial time uniform polynomial-size family of quantum circuits of depth $d(n)$ over arbitrary one-, two-qubit gates and unbounded fan-out gates. Then there exists a polynomial time uniform family $\{\widetilde U(n)\}\subseteq \CRAS$ such that:
\begin{enumerate}[(i)]
    \item $\widetilde U(n)$ realizes the same completely positive trace-preserving map as $U(n)$;
    \item the number of adaptive rounds of $\widetilde U(n)$ is $\mathcal{O}(d(n))$.
\end{enumerate}
In particular, every $\mathcal{O}(1)$-depth fan-out family can be compiled into $\CRAS$.
\end{lemma}

\begin{proof}
Without loss of generality, we can write 
\begin{align}
    U(n)=U_t(n)U_{t-1}(n)\cdots U_1(n),
\end{align}
with $t=d(n)$. Each layer $U_j(n)$ consists of pairwise disjoint elementary gates. Every one-qubit or two-qubit gate is left unchanged. Every fan-out gate is replaced by the corresponding dynamic gadget from Lemma~\ref{fact:BW}. Since the supports of the gates in one quantum circuit layer $U_j(n)$ are disjoint, all substituted gadgets in that layer can be run in parallel after allocating disjoint ancilla blocks.

Each quantum circuit layer $U_j(n)$ therefore expands to only $\mathcal{O}(1)$ adaptive rounds and $\mathcal{O}(1)$ intra-round quantum depth. The ancilla overhead remains polynomial because each gadget uses only linearly many ancillas in its arity and the total fan-out arity is at most the source circuit size. Finally, Lemma~\ref{fact:BW} states that each dynamic gadget is logically equivalent to the original fan-out or CNOT-ladder operation on the data qubits, so every layer $U_j(n)$ and its compiled version implement the same channel on the logical qubits. Composing over all layers proves the claim.
\end{proof}

\emph{Added Note:}
Throughout the appendices, all approximate coherent quantum gates are understood to be instantiated with sufficiently small inverse-polynomial error. More precisely, if a circuit invokes $M(n)=\operatorname{poly}(n)$ such primitives, we choose the clean-ancilla channel error of each invocation to be at most $\varepsilon(n)/M(n)$, where $\varepsilon(n)$ is an arbitrarily small inverse polynomial. A standard argument then bounds the error of the complete circuit by $\varepsilon(n)$, and hence bounds the total-variation distance between its output distribution and the ideal one by the same quantity. When an underlying construction is stated in terms of failure probability rather than error in terms of operator norm, it is amplified sufficiently strongly before being used coherently. These accuracy choices do not change the claimed asymptotic depth or polynomial-size bounds. In particular, whenever an ideal trial succeeds with probability $\Omega(1/n)$, we take $\varepsilon(n)=o(1/n)$, so that the implemented trial retains success probability $\Omega(1/n)$; the parallel repetition used below then yields the stated overall success probability.

\section{Computational Power Hierarchy of $\CRAS$}
\label{app:feedback-hierarchy}

In this section, we note that each such ancilla is measured at most once, and all intermediate outcomes available to the classical controller arise from these measurements. Meanwhile, we require the considered $\CRAS$ quantum circuit receives a quantum input.

\begin{proof}[Proof of Theorem~\ref{thm:hierarchy}]
The inclusion follows immediately from monotonicity in the number of
measurement outcomes. Set $T=m(n)$, take $T+1$ input qubits
$q_0,q_1,\ldots,q_T$, and define an observable $O_T=X_{q_0}\otimes X_{q_1}\otimes \cdots\otimes X_{q_T}$ and arbitrary input quantum state $\rho$, we consider to implement a parity function $F_T=\operatorname{Tr}[O_T\rho]$.

Measure $X_{q_i}$ in parallel for $i=1,\ldots,T$, obtaining
measurement results $s_i\in\{\pm1\}$, and let $s=\prod_{i=1}^Ts_i$.
Apply $Z_{q_0}$ if $s=-1$, then apply $H_{q_0}$ and perform the
final $Z$ measurement. The correction multiplies the effective
$X_{q_0}$ outcome by $s$, hence the expectation value of the final outcome is given by $\mathbb E[Z_{\mathrm{out}}]
=\operatorname{Tr}(O_T\rho).$ Thus $F_T$ is implemented exactly with $T=m(n)$ raw intermediate
outcomes and some universal constant coherent depth $d$.  Since
$m$ is admissible, this construction is uniform and has polynomial
size.

For an input observable $B$, define its Pauli weight
$\mathbf{w}(B)$ as the largest weight of a Pauli string having
a nonzero coefficient in the Pauli expansion of $B$. Let $\mathcal C$ be the CPTP map implemented by the entire CRAS
protocol after averaging over all intermediate measurement outcomes.
Since the final output qubit is measured in the $Z$ basis, define
the effective input observable by $A:=\mathcal C^\dagger(Z)$, such that $\mathbb E[Z_{\mathrm{out}}]
=\operatorname{Tr}\!\left[Z\mathcal C(\rho)\right]
=\operatorname{Tr}[A\rho].$ For each possible sequence of intermediate measurement outcomes, the
feedforward strategy determine a fixed sequence of quantum gates.
In the backward Heisenberg picture, nontrivial support can originate
only from the final one-qubit observable and from the $M$ selected
single-qubit measurement effects.  These give at most $M+1$ support
candidate qubits. A candidate qubit can reach at most $2^d$ input qubits through depth
$d$ of two-qubit gates, so every branch observable has Pauli weight
at most $(M+1)2^d$. Summing over transcripts, averaging over
classical randomness, and contracting the input-independent work
register cannot increase Pauli degree. Therefore
$\mathbf{w}(A)\le (M+1)2^d$.

Suppose now that the protocol approximates $F_T$ uniformly with
error $\varepsilon<1$.  Then
$\lVert A-O_T\rVert_\infty\le\varepsilon$.  If
$\mathbf{w}(A)<T+1$, the coefficient of the weight-$(T+1)$
Pauli string $O_T$ in $A$ is zero. Pauli orthogonality would then
give
\[
1
=
\left|
2^{-(T+1)}
\operatorname{Tr}\!\left[O_T(O_T-A)\right]
\right|
\le
\lVert O_T-A\rVert_\infty
\le
\varepsilon,
\]
which is impossible.  Hence
$\mathbf{w}(A)\ge T+1$, and consequently the relationship $(M+1)2^d\ge T+1$ holds, equivalently, 
\begin{align}
    M\ge 2^{-d}(T+1)-1=\Omega(T)
\end{align}
when the circuit depth $d$ is a constant.
\end{proof}

\section{Proof Sketch of Theorem~\ref{thm:main-maintext}}
\label{proofsketchtheorem2}

The construction can be viewed as an adaptive shallow compression of the
algebraic core of Shor's discrete-logarithm algorithm~\cite{Shor1994,Shor1997}. Technically, the shallow implementation uses the constant-depth power of
unbounded fan-out. Fan-out circuits support arithmetic and Fourier-transform
primitives, including threshold-type operations and approximate quantum Fourier
transforms for arbitrary moduli~\cite{HoyerSpalek2005,BrowneKashefiPerdrix2011,Rosenbaum2013}.
Recent adaptive quantum circuit constructions show that fan-out and related long-range
CNOT operations can be implemented in constant depth using mid-circuit
measurements and classical feedforward~\cite{BaeumerPRXQ2024,BaeumerPRL2024,BaeumerWoerner2025}.
We combine these ingredients with an explicit arithmetic compilation tailored
to the discrete-logarithm reduction.

\begin{theorem}[Formal version of Theorem~\ref{thm:main-maintext}]
Assume that the DLP is not in $\mathsf{BPP}$ in the worst case. 
Then there exists an explicit family of constant-round
adaptive shallow circuits $\{C_{p,g,y}\}_{(p,g,y)\in\mathcal I}\subseteq \CRAS$ such that no randomized
classical polynomial-time algorithm can, on every valid input $(p,g,y)$,
estimate $\mu_{p,g,y}
={\rm Tr}[Z_{{\rm out}}C_{p,q,y}(|0^n\rangle\langle 0^n|)]$ to additive error $1/100$ with success probability at least $2/3$.
\end{theorem}

The proof proceeds by constructing an explicit family of adaptive shallow circuits $C_z\in\CRAS$ whose fixed local readout encodes enough information to recover the discrete logarithm of the input instance $z=(p,g,y)$. At a high level, the construction has three steps. First, we implement, in constant depth with mid-circuit measurements and classical feedforward, the algebraic core of Shor's discrete-logarithm algorithm. Second, the resulting measurement information is converted into a geometric overlap parameter. Third, that parameter is encoded into the quantum mean value of a single output qubit, as illustrated in Fig.~\ref{fig:dlp}. Consequently, if one could classically estimate this one-qubit expectation value to constant additive error on every valid input, then one could solve the prime-field discrete logarithm problem in classical polynomial time.

\begin{figure}[t]
    \centering
    \includegraphics[width=0.9\linewidth]{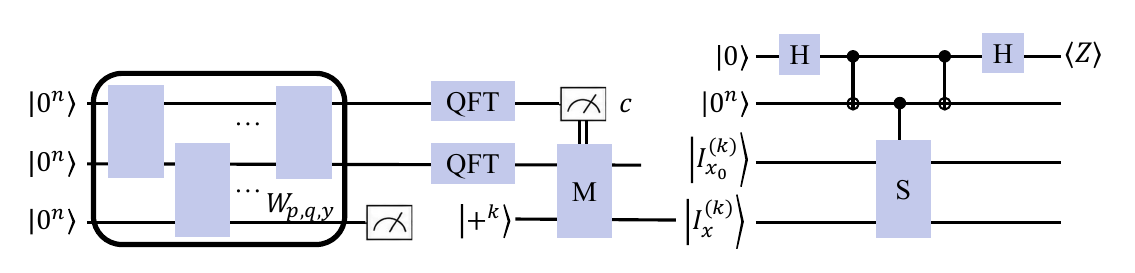}
    \caption{The figure depicts a high-level quantum circuit diagram for encoding the $\rm DLP$ into the $\CRAS$ model. The model consists of a constant number of quantum rounds, each with $\mathcal{O}(1)$ coherent depth, followed by mid-circuit measurements and polynomial-time classical feedforward. The illustrated $\CRAS$ construction starts from a quantum circuit $W_{p,g,y}$ that generates the quantum state $|\phi\rangle$ in Eq.~\eqref{Eq:phistate}, followed by measuring the last register. Then applying $\rm QFT$ circuits and measure the first register. The $\rm M$ represents the multiplication gate such that $|d\rangle|i\rangle\mapsto |d\rangle|i-c^{-1}d\rangle$. Finally the quantum circuit ends up with the generalized swap-test.}
    \label{fig:dlp}
\end{figure}

Our shallow quantum subroutine follows quantum circuit construction underlying Shor's discrete-logarithm algorithm~\cite{Shor1994}. It prepares uniform superpositions over two exponent registers and coherently computes $g^a y^{-b}\pmod p$ for $a,b\in\Z_r$, thereby preparing the state
\begin{align}
   |\phi\rangle
   =\frac{1}{r}\sum_{a,b\in\Z_r}\ket{a}\ket{b}\ket{g^ay^{-b}}
   =\frac{1}{r}\sum_{a,b\in\Z_r}\ket{a}\ket{b}\ket{g^{a-xb}}.
   \label{Eq:phistate}
\end{align}
Equivalently, we construct a quantum circuit $W_{p,g,y}$ that implements the map $|a\rangle|b\rangle|0^n\rangle \mapsto |a\rangle|b\rangle|g^a y^{-b}\bmod p\rangle$. The construction of $W_{p,g,y}$ is the main technical ingredient of the proof. We show that this map can be compiled into a constant-depth adaptive quantum circuit. At a high level, it involves only addition, iterated multiplication, and division with remainder on $\mathcal O(\log p)$-bit integers. These arithmetic tasks can be reduced to constant-depth circuits built from the quantum threshold gate
\[
{\rm TH}_{m,\tau}:\ket{x_1,\ldots,x_m}\ket{b}\mapsto
\ket{x_1,\ldots,x_m}\ket{b\oplus \ind{\sum_i x_i\ge \tau}},
\]
where $m=\poly(n)$ and $\tau\in\mathbb Z$. 

The threshold gate is the key parallel primitive because the main bottleneck in multiplication is carry resolution. Indeed, multiplying two $n$-bit integers first generates $n$ shifted partial-product rows, and the difficulty is then to determine, for each output position, how many lower-order contributions overflow into that position. Without threshold gates, this information must be propagated through a cascade or tree of adders, which leads to $\mathcal O(\log n)$ circuit depth. Threshold gates remove this sequential dependence by testing, in parallel, whether the total contribution to a given position exceeds the relevant carry threshold. In this way, carries for many bit positions can be determined simultaneously, which is the essential mechanism behind constant-depth multiplication in the fan-out model. By Ref.~\cite{HoyerSpalek2005}, threshold gates admit constant-depth constructions in the unbounded-fan-out model. Recent work further shows that unbounded fan-out can be implemented exactly in constant depth using mid-circuit measurements and classical feedforward~\cite{BaeumerWoerner2025}. Combining these ingredients yields the required constant-depth adaptive implementation of $W_{p,g,y}$. We leave technical details to Appendix~\ref{App:Multiplier}.

Next, we measure the third register of $|\phi\rangle$. If the outcome is $g^s$, with $s\equiv a-xb\pmod r$, then the first two registers collapse to the coset state
\begin{align}
    \ket{L_s}=\frac{1}{\sqrt r}\sum_{b\in\Z_r}\ket{s+xb}\ket{b}.
\end{align}
Applying quantum Fourier transformation ${\rm QFT}_r\otimes{\rm QFT}_r$ to this state and only measuring the first register yields $c$, and the remaining quantum state is denoted by $|d\rangle=|-xc \pmod r\rangle$. Whenever $c$ is invertible modulo $r$, one can obtain $c^{-1}\pmod r$. This step is again compatible with constant depth. Approximate quantum Fourier transforms for arbitrary moduli admit constant-depth implementations in the fan-out model~\cite{HoyerSpalek2005}, and fan-out gates themselves can be compiled into dynamic constant-depth circuits using measurement-based constructions~\cite{BaeumerWoerner2025}. Although a single run succeeds only with inverse-logarithmic probability, $\mathcal O(n\log n)$ independent runs can be performed in parallel on disjoint qubit blocks.

Let $k=n-3$ and $L=2^k$, we then prepare a quantum state 
\begin{align}
    |d\rangle|I_0\rangle=\frac{1}{\sqrt{L}}\sum\limits_{i=0}^{L-1}|d\rangle|i\rangle.
\end{align}
Applying the multiply gate $\rm M$, one achieves the map $|d\rangle|i\rangle\mapsto |d\rangle|i-c^{-1}d\rangle$, that is 
\begin{align}
    |d\rangle|I_0\rangle\mapsto |d\rangle|I^{(k)}_x\rangle=|d\rangle\frac{1}{\sqrt{L}}\sum\limits_{i=0}^{L-1}|x+i~{\rm mod}~r\rangle.
\end{align}
Let $x_0=(p+1)/2$, we prepare a quantum state $|I^{(k)}_{x_0}\rangle$, and finally utilize the swap-test to compute their overlap. Here, we note that the swap-test can be efficiently implemented by using Bell measurement. 

For $1\le k\le n-1$ and $u\in \Z_r$, define the cyclic interval
\[
I_u^{(k)}:=\{u,u+1,\ldots,u+2^k-1\}\subseteq \Z_r,
\]
and we observe the quantum state overlap
\begin{align}
K_0(x,x_0)=\abs{\langle I_x^{(k)}|I^{(k)}_{x_0}\rangle}^2=\frac{\bigl|I_x^{(k)}\cap I_{x_0}^{(k)}\bigr|^2}{2^{2k}}.
\end{align}
When $k=n-3$, then $I_x^{(n-3)}=\{x,x+1,\cdots,x+2^{n-3}-1\}$. Since $2^{n-1}<p\leq 2^n$, the interval $I_x^{(n-3)}$ has length $2^{n-3}$, which is a constant fraction of the exponent circle $\mathbb Z_r$. More precisely, its length has the order $\mathcal O(p/8)$. Similarly, $I_{x_0}^{(n-3)}$ is a fixed directed interval starting at $x_0=(p+1)/2$. This overlap quantity separates two promised regions of the exponent space. If $x\in [1,(p-1)/16]$, these two intervals do not overlap, and hence $K_0(x,x_0)=0$. If $x\in[(p-1)/2+1,(p-1)/2+(p-1)/16]$, intersection of the corresponding intervals is at least $p/16$, resulting in $K_0(x,x_0)\geq 1/16$. As a consequence, additive approximation of $K_0(x,x_0)$ to constant error already solves a promise version of prime-field discrete logarithm, determining whether $x\in [1,(p-1)/16]$ or $x\in[(p-1)/2+1,(p-1)/2+(p-1)/16]$ which is polynomial-time equivalent to full prime-field discrete logarithm~\cite{blum2019generate,Liu2021}. We give the rigorous proof in Appendix~\ref{App:proof}.

\section{Shallow-depth adaptive quantum circuit for DLP}
\label{App:proof}
In the Appendix~\ref{sec:Wpgy-composition}, we will theoretically demonstrate how to compile the quantum circuit $W_{p.g.y}$ into a constant-depth quantum circuit. Here, we first accept this result and assume we have obtained the quantum state $|\phi\rangle=\frac{1}{r}\sum_{a,b\in\Z_r}\ket{a}\ket{b}\ket{g^{a-xb}}$, and delay technical details to Appendix~\ref{sec:Wpgy-composition}.

\subsection{Compute bit-string $c$ by QFT}

\begin{lemma}
\label{lem:singletrial}
Let $x=\log_g y\in\Z_r$. Consider a quantum process as follows:
\begin{enumerate}[(i)]
    \item Prepare the uniform superposition over $(a,b)\in(\Z_r\times \Z_r)$;
    \item Coherently compute $g^a y^{-b}\bmod p$ via quantum process $W_{p,g,y}$, and measure the last register in the computational basis;
    \item Apply ${\rm QFT}_r\otimes {\rm QFT}_r$ to the first two registers;
    \item Measure the first register in the computational basis.
\end{enumerate}
Then the output classical bit-string $c\in\Z_r$ and the remaining quantum state $|d\rangle$ satisfy the condition $\{(c,d): d\equiv -xc \pmod r\}$. Moreover, the probability on $\gcd(c,r)=1$ is exactly $\varphi(r)/r$, where $\varphi(r)$ denotes the Euler function.
\end{lemma}

\begin{proof}
Apply the $W_{p,g,y}$ quantum circuit, we prepare the quantum state
\[
|\phi\rangle=\frac{1}{r}\sum_{a,b\in\Z_r}\ket{a}\ket{b}\ket{g^{a-xb}}.
\]
For the sake of analysis, imagine measuring the third register first. If the outcome is $g^z$, then the first two registers collapse to the uniform superposition over the affine line
\[
L_z:=\{(a,b)\in\Z_r^2: a-xb\equiv z\pmod r\},
\]
and the resulting quantum state is given by
\[
\ket{L_z}=\frac{1}{\sqrt r}\sum_{b\in\Z_r}\ket{z+xb}\ket{b}.
\]
Now apply $\rm QFT_r\otimes QFT_r$. Using the fact 
\[
{\rm QFT}_r\ket{u}
=
\frac{1}{\sqrt r}\sum_{c\in\mathbb Z_r}\omega_r^{cu}\ket{c},
\qquad
\omega_r=e^{2\pi i/r},
\]
we obtain
\begin{align}
({\rm QFT}_r\otimes {\rm QFT}_r)\ket{L_z}
&=
\frac{1}{\sqrt r}\sum_{b\in\mathbb Z_r}
\left(
\frac{1}{\sqrt r}\sum_{c\in\mathbb Z_r}\omega_r^{c(z+xb)}\ket{c}
\right)
\otimes
\left(
\frac{1}{\sqrt r}\sum_{d\in\mathbb Z_r}\omega_r^{db}\ket{d}
\right) \\
&=
\frac{1}{r\sqrt r}
\sum_{c,d\in\mathbb Z_r}
\omega_r^{cz}
\left(
\sum_{b\in\mathbb Z_r}\omega_r^{(cx+d)b}
\right)
\ket{c}\ket{d}.
\end{align}

The amplitude of an output pair $(c,d)$ is thus given by
\begin{align}
\frac{1}{r\sqrt r}\sum_{b\in\Z_r}\omega_r^{c(z+xb)+db}
=
\frac{\omega_r^{cz}}{r\sqrt r}\sum_{b\in\Z_r}\omega_r^{(cx+d)b}.
\end{align}
The final sum is a complete geometric series. It equals $r$ if $cx+d\equiv 0\pmod r$ and vanishes otherwise. Hence the output distribution is supported exactly on the line
\begin{align}
    d\equiv -xc\pmod r.
\end{align}
For each valid pair the amplitude magnitude is $1/\sqrt r$, so the valid pairs are equiprobable and the marginal distribution of $c$ is uniform on $\Z_r$.

If $\gcd(c,r)=1$, then $c$ has a unique inverse modulo $r$, denoted $c^{-1}$.
Therefore the success probability of one ideal trial is exactly the fraction of invertible residues modulo $r$, namely $\varphi(r)/r$.
\end{proof}

\begin{lemma}
\label{lem:totient}
There exists an absolute constant $\alpha>0$ such that for every integer $r\ge 3$, the relationship
\begin{align}
    \frac{\varphi(r)}{r}\ge \frac{\alpha}{\log r}
\end{align}
holds.
%Consequently, if $n=|\langle p,g,y\rangle|$, then for every valid instance $(p,g,y)$ with sufficiently large $n$,
%\[
%\frac{\varphi(p-1)}{p-1}\ge \frac{\beta}{n}.
%\]
\end{lemma}

\begin{proof}
Euler's product formula gives
\begin{align}
    \frac{\varphi(r)}{r}=\prod_{q\mid r}\left(1-\frac{1}{q}\right),
\end{align}
where the product ranges over the distinct prime divisors of $r$. Since each factor lies in $(0,1)$, multiplying over the larger set of all primes $q\le r$ can only decrease the product, resulting in
\[
\frac{\varphi(r)}{r}\ge \prod_{q\le r}\left(1-\frac{1}{q}\right).
\]
By Mertens' theorem~\cite{RosserSchoenfeld1962}, the right-hand side is $\Theta(1/\log r)$ as $r\to\infty$. This proves the claim for some absolute constant $\alpha>0$.
\end{proof}

\subsection{Classical Hardness Result of local observable expectation}

\begin{lemma}[\cite{Liu2021}]
\label{fact:Liu}
Fix a prime $p$, a generator $g\in\Z_p^\ast$ and an integer $1\le k\le n-1$, where $n=\lceil\log_2 p\rceil$. When $k=n-3$ and $x_0={(p+1)/2}$, additive $0.01$ approximation of $K_0^{(p,g,k)}(x,x_0)$ solves $\mathrm{DLP}_{1/16}(p,g)$.
\end{lemma}

\begin{proposition}
\label{prop:observable}
For every fixed constant $c>0$, there exists an explicit adaptive shallow-depth quantum circuit family $C_{p,g,y}^{(c)}\in\CRAS$
and one output qubit $q_{\rm out}$ such that
\[
\left|
\bigl\langle Z_{q_{\rm out}}\bigr\rangle_{C_{p,g,y}^{(c)}}
-K_0^{(p,g,k)}(x,x_0)
\right|
\le n^{-c}.
\]
\end{proposition}

\begin{proof}
Apply the generalized swap-test to quantum states $|I^{(k)}_x\rangle$ and $|I^{(k)}_{x_0}\rangle$ naturally proves the statement. More specifically, we initialize the quantum system by
\begin{align}
    |0\rangle_c|0^n\rangle_f|I^{(k)}_x\rangle|I^{(k)}_{x_0}\rangle,
\end{align}
then perform $H$ gate and unbounded fan-out on the system $c$ and $f$, followed by the parallel controlled-swap operator:
\begin{equation}
    \begin{split}
        |0\rangle_c|0^n\rangle_f|I^{(k)}_x\rangle|I^{(k)}_{x_0}\rangle&\mapsto \frac{1}{\sqrt{2}}\left(|0\rangle_c|0^n\rangle_f+|1\rangle_c|1^n\rangle_f\right)|I^{(k)}_x\rangle|I^{(k)}_{x_0}\rangle\\
        &\mapsto\frac{1}{\sqrt{2}}\left(|0\rangle_c|0^n\rangle_f|I^{(k)}_x\rangle|I^{(k)}_{x_0}\rangle+|1\rangle_c|1^n\rangle_f|I^{(k)}_{x_0}\rangle|I^{(k)}_x\rangle\right).
    \end{split}
\end{equation}
Then perform unbounded fan-out on the system $c$ and $f$ and $H$ gate, the quantum state becomes to
\begin{align}
    \frac{1}{2}|0^n\rangle_f\left(|0\rangle_c(I+S)|I^{(k)}_x\rangle|I^{(k)}_{x_0}\rangle+|1\rangle_c(I-S)|I^{(k)}_x\rangle|I^{(k)}_{x_0}\rangle\right).
\end{align}
Measure the $c$ system by Pauli-Z basis yielding the quantum state overlap. We note that all local swap operators can be implemented in parallel, and the unbounded fan-out gate has a constant-depth with mid-circuit measurement implementation.
\end{proof}

We finally combine Proposition~\ref{prop:observable} with Lemma~\ref{fact:Liu}.

\begin{theorem}
\label{thm:main-supp}
Assume that prime-field discrete logarithm is not in $\BPP$ in the worst case. Then there exists an explicit family $\{C_{p,g,y}\}_{(p,g,y)\in\mathcal I}\subseteq\CRAS$ and one fixed one-local observable $Z$ on one output qubit such that no randomized classical polynomial-time algorithm can, on every valid input $(p,g,y)\in\mathcal I$, estimate $\mu_{p,g,y}:=\bigl\langle Z_{q_{\rm out}}\bigr\rangle_{C_{p,g,y}}$ to additive error $1/100$.
\end{theorem}

\begin{proof}
Suppose for contradiction that there exists a randomized classical polynomial-time algorithm $\mathcal{A}$ which, on every valid input $(p,g,y)$, outputs $\widetilde\mu$ satisfying $\abs{\widetilde\mu-\mu_{p,g,y}}\le 1/100$ with success probability at least $2/3$.

Now specialize to the promise problem from Lemma~\ref{fact:Liu}. For parameters $k=n-3$ and $x_0={(p+1)/2}$, the two promised cases in $\rm DLP_{1/16}(p,g)$ satisfy
\[
K_0^{(p,g,k)}(x,x_0)=0
\qquad\text{or}\qquad
K_0^{(p,g,k)}(x,x_0)\ge \frac{1}{16}.
\]
By Proposition~\ref{prop:observable}, this naturally implies $\mu_{p,g,y}\leq 10^{-3}$ in the first case (for some choices of $c$), and
\[
\mu_{p,g,y}\ge \frac{1}{16}-10^{-3}
\]
in the second. These two intervals are separated by much more than $1/100$, so an additive-$1/100$ estimator distinguishes them.

Hence $\mathcal A$ yields a randomized classical polynomial-time algorithm for $\mathrm{DLP}_{1/16}(p,g)$. By Lemma~\ref{lemma:DLPReduction}, this implies a randomized classical polynomial-time algorithm for full prime-field DLP, contradicting the assumption.
\end{proof}

\section{Compilation of the $W_{p,g,y}$ block}
\label{sec:Wpgy-composition}

Let $p$ be a fixed $n$-bit prime and let
$g,y$ be fixed classical parameters. The target is to construct a quantum circuit to achieve the function
\begin{equation}
W_{p,g,y}:\ |a\rangle |b\rangle |z\rangle |0\rangle
\longmapsto
|a\rangle |b\rangle
|z\oplus g^a y^{-b}\bmod p\rangle |0\rangle,
\label{eq:Wpgy-map}
\end{equation}
where $a,b\in\{0,1\}^n$, represented by $a=\sum_{j=0}^{n-1}a_j2^j$ and $b=\sum_{j=0}^{n-1}b_j2^j$. To construct the target quantum circuit, we first classically precompute the constants
\begin{equation}
u_j=g^{2^j}\bmod p,
\qquad
v_j=y^{-2^j}\bmod p,
\qquad 0\le j<n.
\label{eq:Wpgy-precompute}
\end{equation}
Then define $2n$ selected $n$-bit factors by
\begin{equation}
h_j(a_j)=
\begin{cases}
1, & a_j=0,\\
u_j, & a_j=1,
\end{cases}
\qquad
h_{n+j}(b_j)=
\begin{cases}
1, & b_j=0,\\
v_j, & b_j=1,
\end{cases}
\qquad 0\le j<n.
\label{eq:Wpgy-selected-factors}
\end{equation}
For each output bit position of each factor, the loaded bit is one of
$0,1,a_j,\neg a_j$ or one of $0,1,b_j,\neg b_j$, determined by the
fixed pair of constants being selected. Hence all factors can be loaded
coherently in constant depth, using $\mathcal{O}(n^2)$ factor-register qubits and
$\mathcal{O}(n^2)$ fan-out targets.

The product of these selected factors naturally satisfies
\begin{equation}
\prod_{\ell=0}^{2n-1}h_\ell
\equiv
\prod_{j=0}^{n-1}g^{a_j2^j}
\prod_{j=0}^{n-1}y^{-b_j2^j}
=
g^a y^{-b}
\pmod p.
\label{eq:Wpgy-product-correctness}
\end{equation}
As an ordinary integer product, since each factor is smaller than $2^n$,
we have $0\le \prod_{\ell=0}^{2n-1}h_\ell < 2^{2n^2}.$
Thus we set $L=2n^2$.

Let
$U_{\mathrm{MUL}}(m,n)$ denote a coherent iterated-multiplication
subroutine for $m$ many $n$-bit factors, writing their ordinary integer
product into an output register of length $mn$.  In the present
application we use $m=2n$ and output length $L=2n^2$:
\begin{equation}
|h_0,\ldots,h_{2n-1}\rangle |0^L\rangle 
\longmapsto
|h_0,\ldots,h_{2n-1}\rangle
\left|\prod_{\ell=0}^{2n-1}h_\ell\right\rangle.
\label{eq:IMUL-macro}
\end{equation}
We will provide explicit constructions of the quantum gate $U_{\mathrm{MUL}}(m,n)$ in the following sections.

Let $X_{a,b}:=\prod_{\ell=0}^{2n-1} h_\ell(a,b)=\prod_{l=0}^{2n-1}h_l$, and the full clean computation can be written as
\begin{equation}
\begin{aligned}
&|a\rangle |b\rangle |z\rangle
|0\rangle_{\mathsf H}
|0^L\rangle_{\mathsf P}
|0\rangle_{\mathsf A_{\rm I}}
|0\rangle_{\mathsf A_{\rm red}}
\\
\xrightarrow{\,U_{\rm load}\,}\quad
&|a\rangle |b\rangle |z\rangle
|h_0(a,b),\ldots,h_{2n-1}(a,b)\rangle_{\mathsf H}
|0^L\rangle_{\mathsf P}
|0\rangle_{\mathsf A_{\rm I}}
|0\rangle_{\mathsf A_{\rm red}}
\\
\xrightarrow{\,U_{\rm MUL}\,}\quad
&|a\rangle |b\rangle |z\rangle
|h_0(a,b),\ldots,h_{2n-1}(a,b)\rangle_{\mathsf H}
|X_{a,b}\rangle_{\mathsf P}
|0\rangle_{\mathsf A_{\rm I}}
|0\rangle_{\mathsf A_{\rm red}}
\\
\xrightarrow{\,R_{p}\,}\quad
&|a\rangle |b\rangle
|z\oplus (X_{a,b}\bmod p)\rangle
|h_0(a,b),\ldots,h_{2n-1}(a,b)\rangle_{\mathsf H}
|X_{a,b}\rangle_{\mathsf P}
|0\rangle_{\mathsf A_{\rm I}}
|0\rangle_{\mathsf A_{\rm red}}
\\
\xrightarrow{\,U_{\rm MUL}^{\dagger}\ U_{\rm load}^{\dagger}\,}\quad
&|a\rangle |b\rangle
|z\oplus (g^a y^{-b}\bmod p)\rangle
|0\rangle_{\mathsf H}
|0^L\rangle_{\mathsf P}
|0\rangle_{\mathsf A_{\rm I}}
|0\rangle_{\mathsf A_{\rm red}} .
\end{aligned}
\label{eq:Wpgy-flow}
\end{equation}
Here $\mathsf H$ denotes the factor registers holding 
values $h_\ell(a,b)$, $\mathsf P$ is the $L$-bit product accumulator,
$\mathsf A_{\rm I}$ is the ancilla by the iterated-multiplication subroutine, and $\mathsf A_{\rm red}$ is the ancilla required by the fixed-modulus
reduction.  The equality in the last line uses
\begin{equation}
X_{a,b}\bmod p
=
\prod_{\ell=0}^{2n-1}h_\ell(a,b)\bmod p
=
g^a y^{-b}\bmod p .
\end{equation}
In the following sections, we first propose how to reduce the ``division with remainder'' gate $R_{p}$ to standard multiplier gate, and finally propose how to implement the multiplier gate.

\section{Fixed-modulus reduction from multiplication}
\label{sec:fixed-modulus-reduction}

In the context of implementing $W_{p,g,y}$, we only need reduction modulo a fixed classical modulus. Let $p$ be a fixed classical $n$-bit modulus, so $p<2^n$, and $L\ge n$ be a length parameter for the integer to be reduced. We assume $0\le X<2^L$, and the target quantum gate
\[
R_{p}:\ |X\rangle |z\rangle |0\rangle_{\mathsf A_{\rm red}}
\longmapsto
|X\rangle |z\oplus (X\bmod p)\rangle |0\rangle_{\mathsf A_{\rm red}},
\]
where the output residue is represented on $n$ qubits and the modulus
$p$ is fixed throughout the construction.

Given an input $X$, define the classical reciprocal constant $\mu_L=\left\lfloor \frac{2^L}{p}\right\rfloor$, $\hat q=
\left\lfloor \frac{X\mu_L}{2^L}\right\rfloor$ and $q=\left\lfloor \frac{X}{p}\right\rfloor$. Then we have
\[
\frac{X\mu_L}{2^L}
\le
\frac{X}{p}
<
\frac{X\mu_L}{2^L}
+
\frac{X}{2^L}
<
\frac{X\mu_L}{2^L}+1,
\]
where the first two inequalities come from $\mu_L\le \frac{2^L}{p}<\mu_L+1$.
Therefore one has $\hat q\in\{q-1,q\}$.

Now compute $r_0=X-\hat q p$. If $\hat q=q$, then $r_0=X\bmod p$.  If
$\hat q=q-1$, then $r_0=(X\bmod p)+p$.  Hence $0\le r_0<2p$
and the final residue is obtained by a threshold function:
\begin{align}
    X\bmod p
=
r_0-\ind{r_0\ge p}\,p.
\end{align}

The quantum circuit $R_{p}$ can be implemented by the following steps:
\begin{enumerate}
\item Compute products $X\mu_L$ and $\hat{q}p$ by multiplier gate.
\item Compute $r_0=X-\hat q p$ in an $\mathcal O(L)$-bit scratch register.
\item Compute $\gamma=\mathbf{1}[r_0\ge p]$, conditionally subtract
$p$ from $r_0$, and XOR the resulting residue into $z$.
\item Reverse all temporary computations, erasing $\gamma$, the
corrected residue, $r_0$, $\hat q p$, and $\hat q$.
\end{enumerate}
As a result, the fixed modulus gate can be further decomposed by multiplier and threshold gates. Here, we note that $p$ is a fixed classical constant, the involved threshold gate thus admits a
polynomial-size constant-depth decomposition into AND and OR gates~\cite{Vollmer1999}, which are
special cases of the threshold gates in Def.~\ref{def:threshold}.

\section{Multiplier Gate by Constant-Depth Fan-out}
\label{App:Multiplier}
Here, we consider the construction of multiplier gate such that
\begin{align}
    |x\rangle|y\rangle|z\rangle\mapsto|x\rangle|y\rangle|z\xor xy\rangle
\end{align}
in the context of fan-out model. By Refs~\cite{BaeumerPRXQ2024,BaeumerWoerner2025}, a $n$-qubit fan-out gate can be implemented by a constant depth quantum circuit with classical feedforward, we thus utilize the fan-out gate as a fundamental element to construct the multiplier gate. To do this, we first introduce the threshold gate which can be explicitly implemented by fan-out gate, then give a white-box construction on utilizing the threshold gate to simulate the multiplier gate.

An unbounded fan-out gate with one control and $m$ targets is the unitary in achieving
\begin{equation}
\ket{b}\ket{t_1,\ldots,t_m}
\longmapsto
\ket{b}\ket{t_1\xor b,\ldots,t_m\xor b}.
\end{equation}
Here, we note that this gate copies computational-basis information coherently, and it is not an
operation that clones arbitrary quantum states.

%We use the following resource notation.  For a subroutine $\mathcal B$, let $A_{\mathcal B}$, $F_{\mathcal B}$, $T_{\mathcal B}$ and $D_{\mathcal B}$ denote the number of auxiliary qubits, the number of unbounded
%fan-out gates, the total number of fan-out targets, and the fan-out circuit depth, respectively.

\begin{definition}[Threshold Gate]
\label{def:threshold}
For $1\le \tau\le m$, we define the threshold gate
\begin{equation}
\Thr^{\xor}_{m,\tau}:
\ket{v_1,\ldots,v_m}\ket{b}\ket{0^{A_{\rm TH}(m,\tau)}}
\longmapsto
\ket{v_1,\ldots,v_m}
\ket{b\xor \ind{\sum_i v_i\ge \tau}}
\ket{0^{A_{\rm TH}(m,\tau)}},
\end{equation}
where the notation $A_{\rm TH}(m,\tau)$ represents the number of ancilla qubits in the construction. We write $\OR_m=\Thr_{m,1}$ and $\AND_m=\Thr_{m,m}$.
\end{definition}

In the following parts, we attempt to implement quantum circuit in the context of threshold gate which can be decomposed into a constant-depth quantum circuit with one-qubit gates and unbounded fan-out gates, and the corresponding circuit size is $\mathcal{O}(n^2\log(n))$~\cite{HoyerSpalek2005,takahashi2016collapse}. 

\vspace{10px}

\noindent\textbf{Weighted thresholds.} For non-negative integer weights $w_1,\ldots,w_s$, define
\begin{equation}
\WTH_{w,\tau}(v)=\ind{\sum_{i=1}^s w_i v_i\ge \tau}.
\end{equation}
Let $M=\sum_{i=1}^s w_i$ and $r=|\{i:w_i>0\}|$. A weighted threshold can be reduced to an unweighted threshold by replacing each input $v_i$ by $w_i$ copies. The Hamming weight of the resulting list is
exactly $\sum_i w_i v_i$. 

%Using the original wire as one occurrence, and
%creating $w_i-1$ additional occurrences when $w_i>0$, the copy step alone
%contributes
%\begin{equation}
%A_{\rm copy}=\sum_{w_i>0}(w_i-1)=M-r,
%\qquad
%F_{\rm copy}\le 2r,
%\qquad
%T_{\rm copy}=2(M-r),
%\label{eq:wth-copy-cost}
%\end{equation}
%where the factor of $2$ accounts for copying and uncopying.  Therefore the
%resources of a clean weighted-threshold gate are bounded by
%\begin{align}
%A_{\rm WTH}(w,\tau)&\le M-r+A_{\rm TH}(M,\tau),\nonumber\\
%F_{\rm WTH}(w,\tau)&\le 2r+F_{\rm TH}(M,\tau),\nonumber\\
%T_{\rm WTH}(w,\tau)&\le 2(M-r)+T_{\rm TH}(M,\tau),\nonumber\\
%D_{\rm WTH}(w,\tau)&\le 2+D_{\rm TH}(M,\tau).
%\label{eq:wth-resources}
%\end{align}
%This construction is coherent, since the additional copy registers are uncomputed.

\vspace{10px}
\noindent \textbf{Parity.}
The parity operation
\begin{equation}
\ket{x_1,\ldots,x_M}\ket{z}
\longmapsto
\ket{x_1,\ldots,x_M}\ket{z\xor x_1\xor\cdots\xor x_M}
\end{equation}
is constant depth in the fan-out model.  Indeed, CNOT is reversed by Hadamard
conjugation on its two qubits, and therefore an unbounded fan-out gate is
conjugate to an unbounded parity gate.  Thus a parity gate is implemented by a
layer of Hadamards, one fan-out gate, and another layer of Hadamards.

\subsection{High-level construction idea}

The multiplier is a reversible, parallelized version of column-wise
multiplication.  The construction has the following layers.

\begin{enumerate}
\item Compute all local products $p_{ij}=x_i y_j$.
The products in column $k$ are those with $i+j=k$.

\item Define the column population $u_k=\sum_{i+j=k}p_{ij}$. If $d_k$ is the integer carry entering column $k$, then the usual recurrence is
\begin{equation}
d_{k+1}=\left\lfloor \frac{u_k+d_k}{2}\right\rfloor,
\qquad
r_k=(u_k+d_k)\bmod 2.
\end{equation}
A direct ripple-carry implementation is not used.

\item Group the $2n$ columns into logarithmic-length blocks.  Let
$R=\lceil \log(n+1)\rceil+1$, $W=2^R$, and let block $b$
start at column $s_b=bR$.  For a prefix of length $t$ inside block
$b$, define the local weighted prefix sum
\[
A_{b,t}=\sum_{j=0}^{t-1}2^j u_{s_b+j},
\]
with $u_k=0$ for padded columns beyond $2n-1$.  For each such block
prefix, form threshold-table bits
\[
Q_{b,t,\ell}=\ind{A_{b,t}\ge \ell}.
\]
These tables encode the block-prefix sums in unary and let us express
the block-carry and in-block carry predicates by constant-depth Boolean
combinations.

\item Use a set--hold--reset prefix formula to generate all block-entry
carries. Let $\delta_k:=\max_{x,y} d_k$ denote the largest possible integer carry entering column $k$; its
closed form is given in Eq.~\eqref{eq:deltak} below. Then utilizing the threshold table $Q_{b,t,l}$ to generate the
unary carry bits
\[
C_{k,c}=\ind{d_k\ge c},
\qquad 1\le c\le \delta_k .
\]

\item Output the product bit directly by parity:
\begin{equation}
r_k=
\left(\bigoplus_{i+j=k}p_{ij}\right)
\xor
\left(\bigoplus_{c=1}^{\delta_k} C_{k,c}\right).
\label{eq:high-level-output-parity}
\end{equation}
The first parity is $u_k\bmod 2$; the second is $d_k\bmod 2$, since
$C_{k,1},\ldots,C_{k,\delta_k}$ are the unary carry bits for $d_k$.

\item XOR all $r_k$ into $z$, and then reverse the computation to erase every
work register.
\end{enumerate}

The only nontrivial part of the construction is the parallel generation of the
integer carries $d_k$.  This is the purpose of the block construction below.

\subsection{Column form of multiplication}

We write the binary representations of input values $(x,y)$, that is
\begin{equation}
x=\sum_{i=0}^{n-1}x_i2^i,
\qquad
y=\sum_{j=0}^{n-1}y_j2^j,
\end{equation}
where $x_i,y_j\in\{0,1\}$. The local products are defined by $p_{ij}=x_i y_j$. For $0\le k\le 2n-1$, define
\begin{equation}
u_k=\sum_{i+j=k}p_{ij},
\end{equation}
with $u_{2n-1}=0$.  The number of possible local products in column $k$ is
\begin{equation}
m_k=
\begin{cases}
k+1, & 0\le k\le n-1,\\
2n-1-k, & n\le k\le 2n-2,\\
0, & k=2n-1.
\end{cases}
\label{eq:mk}
\end{equation}
Here $m_k$ is a deterministic column width, while $u_k$ depends on the input.
The column recurrence is
\begin{equation}
d_0=0,
\qquad
w_k=u_k+d_k,
\qquad
r_k=w_k\bmod 2,
\qquad
d_{k+1}=\left\lfloor \frac{w_k}{2}\right\rfloor .
\label{eq:carry-recurrence}
\end{equation}
The carry $d_k$ is an integer carry, not a single bit.  This is because a column
of a multiplication table may contain many local products.

The largest possible carry entering column $k$ is
\begin{equation}
\delta_k:=\max d_k=
\begin{cases}
0, & k=0,\\
k-1, & 1\le k\le n,\\
2n-k, & n+1\le k\le 2n-1.
\end{cases}
\label{eq:deltak}
\end{equation}
Indeed, the recurrence for the maximum carry is
\begin{equation}
\delta_{k+1}=\left\lfloor\frac{m_k+\delta_k}{2}\right\rfloor,
\qquad \delta_0=0,
\end{equation}
which gives Eq.~\eqref{eq:deltak}.  In particular, we have
\begin{equation}
\sum_{k=0}^{2n-1}\delta_k=n(n-1).
\label{eq:sum-delta}
\end{equation}

\begin{lemma}[Partial-product layer]
All $n^2$ partial products $p_{ij}$ can be generated coherently in constant
fan-out depth.
\end{lemma}

\begin{proof}
Fan out each $x_i$ to $n$ local copies and each $y_j$ to $n$ local copies.  On
the local pair corresponding to $(i,j)$, apply a clean two-input AND, i.e.
$\Thr^{\xor}_{2,2}$, with target $p_{ij}$.  Then uncopy the local copies.  The
$p_{ij}$ registers are retained until the final uncomputation.
\end{proof}

\subsection{Block carry propagation}

We now replace the ripple-carry chain by a block-prefix computation.  Choose $R=\lceil\log(n+1)\rceil+1$ and $W=2^R$. Then we have $n<W\le 4(n+1)$. Split the $2n$ columns into
\begin{equation}
B=\left\lceil\frac{2n}{R}\right\rceil
\end{equation}
blocks of length $R$.  Blocks are indexed by $b=0,\ldots,B-1$, and the first
column of block $b$ is $s_b=bR$. If the final block extends past column $2n-1$, set $u_k=m_k=0$ for the padded
columns.

For $0\le t\le R$, define the block-prefix sum
\begin{equation}
A_{b,t}=\sum_{j=0}^{t-1}2^j u_{s_b+j}.
\label{eq:Abt}
\end{equation}
The full block sum is $A_b=A_{b,R}$.  We decompose it as
\begin{equation}
A_b=\alpha_b W+\rho_b,
\qquad
0\le \rho_b<W.
\label{eq:alpha-rho}
\end{equation}
The maximum possible value of $A_{b,t}$ is
\begin{equation}
M_{b,t}=\sum_{j=0}^{t-1}2^j m_{s_b+j}.
\label{eq:Mbt}
\end{equation}
The quantity $M_{b,t}$ is not a dynamical register, while it only specifies how many
threshold bits are needed later.  Since $m_k\le n$,
\begin{equation}
M_{b,t}\le n(2^t-1)<4n(n+1).
\label{eq:Mstar-bound}
\end{equation}
Let $a_b=\left\lfloor\frac{M_{b,R}}{W}\right\rfloor$, then we have the relationship $0\le \alpha_b\le a_b\le n-1$.

Let $D_b=d_{s_b}$ be the carry entering block $b$. Processing a full block means adding the
incoming carry to the local block sum and discarding the $R$ output bits inside
the block. Therefore the carry
\begin{equation}
D_{b+1}=d_{s_b+R}=\left\lfloor\frac{A_b+D_b}{W}\right\rfloor.
\label{eq:block-exit-first}
\end{equation}
Using Eq.~\eqref{eq:alpha-rho}, and $D_b\le n<W$, this becomes
\begin{equation}
D_{b+1}=\alpha_b+\ind{\rho_b+D_b\ge W}.
\label{eq:block-exit}
\end{equation}
Thus the locally generated part is $\alpha_b$, and only one extra bit needs to
propagate from one block to the next.

\vspace{8px}
\begin{lemma}[Set--hold--reset carry bit]
Set $\alpha_{-1}=0$ and $E_0=0$.  Write
\begin{equation}
D_b=\alpha_{b-1}+E_b,
\label{eq:Db-alpha-E}
\end{equation}
where $E_{b+1}=\ind{\rho_b+D_b\ge W}$ represents the carry given by $D_{b-1}$. For block $b$, define $X_b=\rho_b+\alpha_{b-1}, S_b=\ind{X_b\ge W}$, and $H_b=\ind{X_b=W-1}$. Then
\begin{equation}
E_{b+1}=
\begin{cases}
1, & S_b=1,\\
E_b, & H_b=1,\\
0, & \text{otherwise}.
\end{cases}
\label{eq:set-hold-reset}
\end{equation}
Equivalently, all $E_{b+1}$ are given by the parallel prefix formula
\begin{equation}
E_{b+1}=\vee_{\ell=0}^{b}
\left(
S_\ell\wedge_{t=\ell+1}^{b}H_t
\right).
\label{eq:prefix-E}
\end{equation}
\end{lemma}

\begin{proof}
Substituting Eq.~\eqref{eq:Db-alpha-E} into Eq.~\eqref{eq:block-exit} gives
\begin{equation}
E_{b+1}=\ind{\rho_b+\alpha_{b-1}+E_b\ge W}=\ind{X_b+E_b\ge W}.
\end{equation}
If $X_b\ge W$, the next bit is forced to $1$.  If $X_b=W-1$, the next bit is
exactly $E_b$.  If $X_b\le W-2$, the next bit is forced to $0$.  This gives
Eq.~\eqref{eq:set-hold-reset}. Expanding the recurrence yields Eq.~\eqref{eq:prefix-E}.  Indeed, the
extra carry entering block $b+1$ is equal to $1$ if and only if there
exists an earlier block $\ell\le b$ that sets the extra carry, and every
block between $\ell+1$ and $b$ holds it.
\end{proof}

\subsection{Using Threshold tables to Compute $S$ and $H$}
Now we only need to demonstrate how to efficiently compute $S_b$ and $E_b$ from $A_{b,t}$. For each block
prefix $A_{b,t}$ we generate the unary table
\begin{equation}
Q_{b,t,\ell}=\ind{A_{b,t}\ge \ell},
\qquad
1\le \ell\le M_{b,t}.
\label{eq:Q-table}
\end{equation}
We also use the conventions $Q_{b,t,\ell}=1 \quad (\ell\le 0)$ and $Q_{b,t,\ell}=0 \quad (\ell>M_{b,t})$. The table entries are bounded-weight thresholds in the partial products: the
largest weight is $2^{R-1}=\mathcal O(n)$, and the largest expanded fan-in is
$M_{b,t}<4n(n+1)$.

From the table $Q_{b,R,\ell}$ we generate the quotient bits
\begin{equation}
{\Omega}_{b,a}=\ind{\alpha_b=a}
=Q_{b,R,aW}\wedge \neg Q_{b,R,(a+1)W},
\label{eq:alpha-bits}
\end{equation}
where $Q_{b,R,aW}=1$ and $Q_{b,R,(a+1)W}=0$ indicates $A_{b,R}\geq aW$ and $A_{b,R}<(a+1)W$. In the above construction, recall that we assume $A_b=A_{b,R}=\alpha_bW+\rho_b$. These results imply $a=\alpha_b$.
Predicates involving $\rho_b=A_b\bmod W$ are not ordinary thresholds in $A_b$, instead, they are unions of intervals. In this simplified construction we do not store a
separate $\rho$ table.  Instead, we compute the set and hold bits directly as
DNF predicates in the $Q$-table and the preceding Alpha bits.  With
$\alpha_{-1}=0$, define $a_{-1}=0$ and treat ${\Omega}_{-1,0}$ as the
constant $1$.  Then
\begin{align}
S_b
&=\vee_{a=0}^{a_{b-1}}
\vee_{r=0}^{a_b}
\left(
{\Omega}_{b-1,a}
\wedge Q_{b,R,rW+W-a}
\wedge \neg Q_{b,R,(r+1)W}
\right),
\label{eq:S-direct}\\
H_b
&=\vee_{a=0}^{a_{b-1}}
\vee_{r=0}^{a_b}
\left(
{\Omega}_{b-1,a}
\wedge Q_{b,R,rW+W-1-a}
\wedge \neg Q_{b,R,rW+W-a}
\right).
\label{eq:H-direct}
\end{align}
Above boolean function can be understood as follows. For example, $S_b=1$ implies $\rho_b\geq W-a$. Suppose $A_b=rW+\rho_b$, then one has $A_b=rW_+\rho_b\geq rW+W-a$. Meanwhile, $A_b<(r+1)W$, as a result, $rW+W-a\leq A_b<(r+1)W$, equivalently, its boolean function can be written by $Q_{b,R,rW+W-a}\wedge \neg  Q_{b,R,(r+1)W}$. Furthermore, it is required to set $\alpha_{b-1}=a$, and this give rises to $S_b$. Similarly, we con write $H_b$ as the boolean function in terms of $\Omega$ and $Q$. Indeed, for a fixed $a$, the first formula tests $\rho_b\ge W-a$, and the
second tests $\rho_b=W-1-a$.

Once the block-entry bit $E_b$ is known, the carry entering any column inside
block $b$ can be recovered from the same threshold table.  Let $k=s_b+t$ with
$0\le t<R$ and $k<2n$.  Then
\begin{equation}
d_k=d_{s_b+t}=\left\lfloor\frac{A_{b,t}+D_b}{2^t}\right\rfloor.
\label{eq:inside-block-carry}
\end{equation}
For $1\le c\le \delta_k$, define the unary carry bit
\begin{equation}
C_{k,c}=\ind{d_k\ge c}.
\label{eq:Ckc}
\end{equation}
For $b\ge 1$, using $D_b=\alpha_{b-1}+E_b$, we compute
\begin{align}
C_{k,c}
=&\vee_{a=0}^{a_{b-1}}
\left(
{\Omega}_{b-1,a}\wedge \neg E_b\wedge Q_{b,t,c2^t-a}
\right)
\nonumber
\vee_{a=0}^{a_{b-1}}
\left(
{\Omega}_{b-1,a}\wedge E_b\wedge Q_{b,t,c2^t-a-1}
\right).
\label{eq:C-general}
\end{align}

Finally, we do not form a unary table for
$u_k+d_k$. Instead, we use the parity function. Since $u_k=\sum_{i+j=k}p_{ij}$
we have
\begin{equation}
u_k\bmod 2=\bigoplus_{i+j=k}p_{ij}.
\end{equation}
Similarly, because $C_{k,c}=\ind{d_k\ge c}$ is the unary representation of
$d_k$, exactly $d_k$ of the bits $C_{k,1},\ldots,C_{k,\delta_k}$ are equal to
one.  Hence
\begin{equation}
d_k\bmod 2=\bigoplus_{c=1}^{\delta_k}C_{k,c}.
\end{equation}
Therefore the product bit is
\begin{equation}
r_k=(u_k+d_k)\bmod 2
=
\left(\bigoplus_{i+j=k}p_{ij}\right)
\xor
\left(\bigoplus_{c=1}^{\delta_k}C_{k,c}\right).
\label{eq:rk-direct}
\end{equation}
All these parity gates have disjoint input sets for different $k$ and can be
applied in parallel.  The resulting bits are XORed directly into the output
register $z$.

\subsection{Threshold Gate Count by Numerical Simulation}
\begin{figure}[t]
    \centering
    \includegraphics[width=0.9\linewidth]{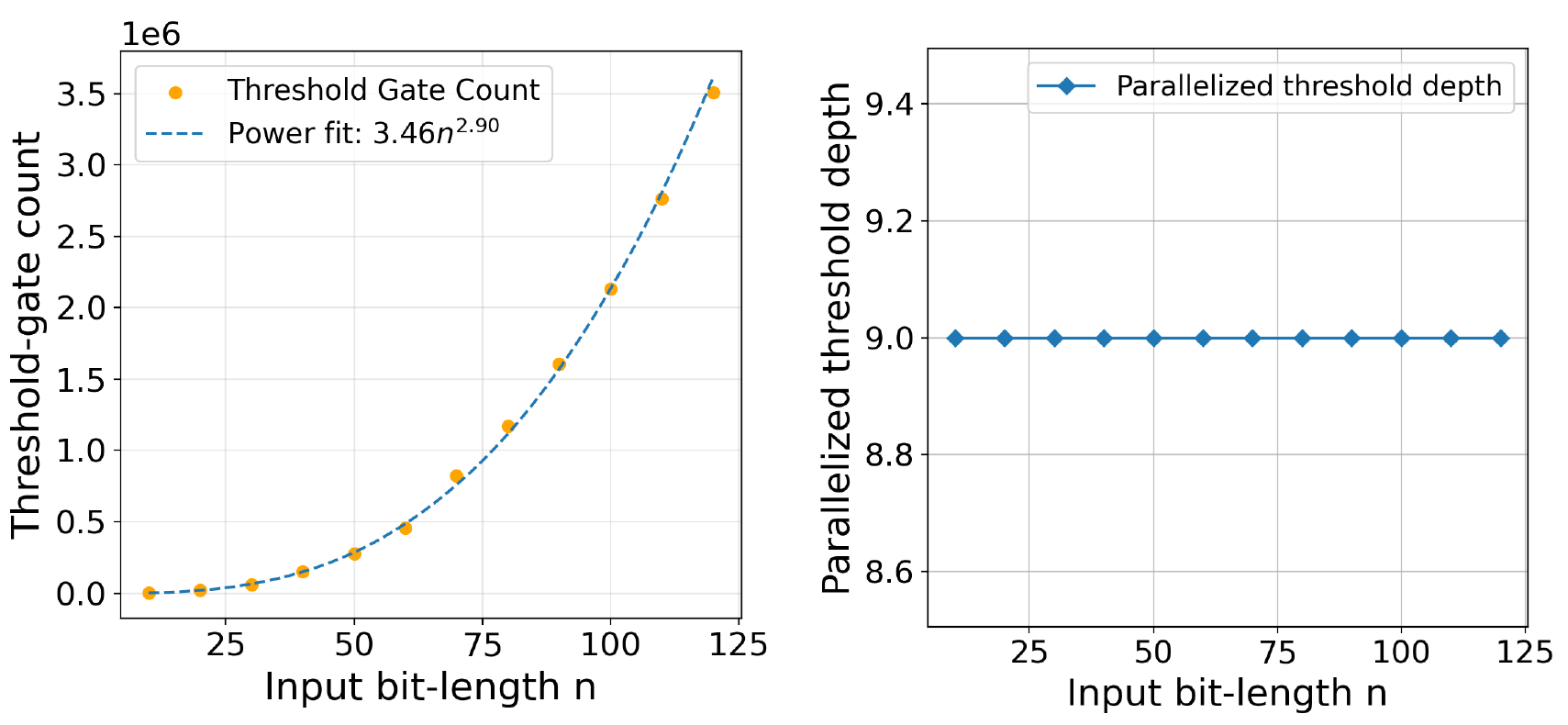}
    \caption{Threshold gate count and depth estimation: We evaluated the threshold-gate count (depth) for
$10\leq n\leq 120$. A least-squares fit of the gate count indicates an approximately $\mathcal O(n^3)$ scaling. Moreover, assuming sufficiently many ancilla qubits, the threshold gates in each stage can be evaluated in parallel, so the resulting circuit depth remains constant.}
    \label{fig:thresholdgate}
\end{figure}

We numerically implemented a resource estimation for the above block-carry multiplier in terms of the threshold-gate count, as shown in Fig.~\ref{fig:thresholdgate}. We do not simulate the quantum state, instead, it follows the logic of our construction. The numerical count is performed at the threshold gate level.  More precisely, we count the threshold gate generated in the stages described above. This includes the two-input predicates for the partial products $p_{ij}=x_i y_j$, the weighted-threshold predicates defining $Q_{b,t,\ell}$, the quotient predicates $\Omega_{b,q}$, the auxiliary predicates used to form $S_b,H_b$ and $E_b$, and the predicates generating the unary carry bits $C_{k,c}$. The final parity layer that writes the output bits is not included in the threshold-gate count. For the depth estimate, we use the same parallelization strategy as in the construction. All predicates within a fixed stage are evaluated in parallel, assuming sufficiently many ancilla qubits. Thus, the multiplier has constant threshold-gate depth.

\section{Generalize to Multi-Input Scenario}
\label{multi}
Here, we briefly discuss how to extend the construction given by Appendix~\ref{App:Multiplier} from two bit-rows to $M=\poly(n)$ bit-rows. Given
$M$ integers $R^{(1)},\ldots,R^{(M)}<2^L$, let
\begin{align}
    L' := L+\lceil \log_2 M\rceil
\end{align}
and pad every row to length $L'$. Write
\begin{align}
    R^{(s)}=\sum_{k=0}^{L'-1} r_k^{(s)}2^k,\qquad
u_k:=\sum_{s=1}^M r_k^{(s)},\qquad
d_0=0,\qquad
d_{k+1}=\Bigl\lfloor \frac{u_k+d_k}{2}\Bigr\rfloor.
\end{align}
Then the $k$th output bit of $\sum_{s=1}^M R^{(s)}$ can be computed by $\eta_k=(u_k+d_k)\bmod 2$. Since $u_k\le M$, one has $d_k\le M-1$ for all $k$. Set
\begin{align}
    R_M:=\lceil\log_2(M+1)\rceil+1,\qquad
W_M:=2^{R_M},
\end{align}
group columns into blocks of length $R_M$, and define the block-prefix sums
\begin{align}
    A_{b,t}:=\sum_{j=0}^{t-1}2^j u_{s_b+j},
\qquad
Q_{b,t,\ell}:=1[A_{b,t}\ge \ell].
\end{align}
Using the same set--hold--reset formulas $S_b,H_b,E_b$ and the same unary carry predicates
$C_{k,c}=1[d_k\ge c]$ as above, with the two-input column-width bounds replaced by the
deterministic bound $M$, all block-entry carries and in-block carries are generated in constant
depth. The output bit is then
\begin{align}
    \eta_k=
\left(\bigoplus_{s=1}^M r_k^{(s)}\right)
\oplus
\left(\bigoplus_{c=1}^{M-1} C_{k,c}\right).
\end{align}
Hence Appendix~\ref{App:Multiplier} yields a uniform polynomial-size constant-depth fan-out circuit
$U_{\mathrm{SUM}}(M,L)$ that coherently XORs $\sum_{s=1}^M R^{(s)}$ into an $L'$-bit target register.

This gives the many-input block $U_{\mathrm{MUL}}(m,n)$ from Appendix G without using a binary
tree of two-input multipliers. Let $h_0,\ldots,h_{m-1}<2^n$ and
\begin{align}
    X:=\prod_{j=0}^{m-1} h_j < 2^{mn}.
\end{align}
Choose pairwise distinct $\mathcal O(\log n)$-bit primes $q_1,\ldots,q_s$ such that
\begin{align}
    Q:=\prod_{\ell=1}^s q_\ell > 2^{mn}.
\end{align}
For each residue channel $q_\ell$, first compute
\begin{align}
    r_{j,\ell}:=h_j \pmod{q_l}
=
\left(\sum_{t=0}^{n-1} h_{j,t}(2^t \bmod q_\ell)\right)\pmod{q_l} ,
\end{align}
so each channel input is obtained by summing selected classical constants and then applying the
fixed-modulus reduction of Appendix~\ref{sec:fixed-modulus-reduction}. If some $r_{j,\ell}=0$, then the channel product is
$R_\ell=0$. Otherwise fix a generator $g_\ell$ of $\mathbb Z_{q_\ell}^{\ast}$, write
\begin{align}
    r_{j,\ell}=g_\ell^{\lambda_{j,\ell}},
\qquad
e_\ell:=\left(\sum_{j=0}^{m-1}\lambda_{j,\ell}\right)\bmod (q_\ell-1),
\end{align}
and set
\begin{align}
    R_\ell = g_\ell^{e_\ell}\bmod q_\ell = X \bmod q_\ell.
\end{align}
Since $q_\ell$ only has $\mathcal O(\log n)$ bits, the maps
$r\mapsto \log_{g_\ell}(r)$ and $e\mapsto g_\ell^e \bmod q_\ell$
are hardwired lookup tables of polynomial size and constant depth, while the sum defining $e_\ell$
is performed by $U_{\mathrm{SUM}}(m,\mathcal O(\log n))$ followed by fixed-modulus reduction.

Finally, with
\begin{align}
    M_\ell:=Q/q_\ell,\qquad N_\ell:=M_\ell^{-1}\pmod{q_\ell},
\end{align}
the unique integer $X<2^{mn}<Q$ satisfying the residue conditions is
\[
X=
\left(\sum_{\ell=1}^s R_\ell M_\ell N_\ell\right)\bmod Q.
\]
Each term summand is a variable times a classical constant, so the final reconstruction is again a
parallel sum of selected classical constants, implemented by $U_{\mathrm{SUM}}$ together with the
fixed-modulus reduction of Appendix~\ref{sec:fixed-modulus-reduction}. Therefore the iterated multiplication block reuses the
Appendix-\ref{App:Multiplier} carry-resolution backend; only the front end is changed. This yields a uniform
polynomial-size constant-depth implementation of $U_{\mathrm{MUL}}(m,n)$ in the fan-out model,
and hence, by Lemma~\ref{lem:compile}, a CRAS implementation.

\section{Proof of Theorem~\ref{lowerbound}}

\begin{theorem}[Sec~3 of Ref~\cite{TopuzogluWinterhof2007}]
\label{thm:brandstatter-winterhof}
Let $F:\mathbb Z\to\{-1,+1\}$ be a periodic sign function with
period $t$, and let $m=\lfloor\log_2t\rfloor$.  For every integer
$0\leq y<2^m$, write
$y=\sum_{i=1}^m2^{i-1}y_i$, where
$(y_1,\ldots,y_m)\in\{0,1\}^m$, and let
$[m]=\{1,\ldots,m\}$.  For each $S\subseteq[m]$, define the
normalized Walsh--Fourier coefficient of the first $2^m$ values of
$F$ by
\[
 \widehat F(S):=
 2^{-m}\sum_{y=0}^{2^m-1}
 F(y)(-1)^{\sum_{i\in S}y_i}.
\]
Define the order-two correlation measure of one period of $F$ by
\[
 C_2(F):=
 \max_{\substack{
       0\leq d_1<d_2<t\\
       1\leq M\leq t-d_2
 }}
 \left|
 \sum_{j=0}^{M-1}
 F(j+d_1)F(j+d_2)
 \right|.
\]
Then
\[
 \max_{S\subseteq[m]}|\widehat F(S)|
 <
 8^{1/4}2^{-m/4}C_2(F)^{1/4}.
\]
\end{theorem}

\begin{proof}[Proof of Theorem~\ref{lowerbound}]
Let $m=\lfloor\log_2p\rfloor=n-1$.  We identify each integer
$0\le y<2^m$ with its binary expansion
$y=(y_1,\ldots,y_m)\in\{0,1\}^m$, where
$y=\sum_{i=1}^m2^{i-1}y_i$.  For every nonzero residue $y$, let
$\operatorname{ind}_g(y)\in\{1,\ldots,p-1\}$ be determined by
$g^{\operatorname{ind}_g(y)}=y\pmod p$.  In particular,
$\operatorname{ind}_g(1)=p-1$.

We use the following sign function on $\{0,\ldots,p-1\}$:
\[
F(y)=
\begin{cases}
+1,
 &y=0\ \text{or}\
  1\le \operatorname{ind}_g(y)\le (p-1)/2,\\
-1,
 &(p+1)/2\le \operatorname{ind}_g(y)\le p-1.
\end{cases}
\]
The value $F(0)=+1$ is fixed independently of the circuit.  In the
Fourier part of the proof, $F$ denotes its restriction to
$0\le y<2^m$.

For this $p$-term sign sequence, let
\[
C_2(F):=
\max_{\substack{0\le d_1<d_2<p\\L\ge1,\;L+d_2\le p}}
\left|
\sum_{t=0}^{L-1}F(t+d_1)F(t+d_2)
\right|.
\]
Apply Gyarmati's theorem~\cite{gyarmati2004family} to
$F(1),\ldots,F(p-1)$, gives
$C_2(F(1),\ldots,F(p-1))
\le320p^{1/2}(\ln p)^3$.
Adding the fixed first value $F(0)$ affects any admissible
correlation sum in at most one summand.  Therefore
\[
C_2(F)\le320p^{1/2}(\ln p)^3+1.
\]

We next introduce the Fourier notation used below.  Write
$[m]=\{1,\ldots,m\}$.  For any real function
$Q:\{0,1\}^m\to\mathbb R$ and $S\subseteq[m]$, let
\[
\widehat Q(S)
:=
2^{-m}\sum_{y=0}^{2^m-1}
Q(y)(-1)^{\sum_{i\in S}y_i},
\qquad
\deg(Q):=
\max\{|S|:\widehat Q(S)\ne0\}.
\]
This Fourier degree is also the degree of the unique multilinear
polynomial representing $Q$ on the Boolean cube. Furthermore, Theorem~\ref{thm:brandstatter-winterhof} gives
\[
\max_{S\subseteq[m]}|\widehat F(S)|
\le
8^{1/4}2^{-m/4}C_2(F)^{1/4}
\le
8^{1/4}2^{-m/4}
\left(320p^{1/2}(\ln p)^3+1\right)^{1/4}.
\]
Since $2^m<p<2^{m+1}$, we have
$p^{1/8}=\Theta(2^{m/8})$ and $\ln p=\Theta(m)$.  Hence
\[
\max_{S\subseteq[m]}|\widehat F(S)|
=
\mathcal O\!\left(2^{-m/8}m^{3/4}\right).
\]

Let $P(y)$ be the signed output bias of the $\CRAS$ on input $y$:
\[
P(y)
:=
\Pr[+1| y]
-
\Pr[-1|y].
\]
Thus $|P(y)|\le1$.  Every $1\le y<2^m$ is a legal group element,
and success probability at least $2/3$ gives
$F(y)P(y)\ge1/3$.  At $y=0$, where no correctness condition is
imposed, we only have $F(0)P(0)\ge-1$.  For the normalized inner
product
$\langle Q_1,Q_2\rangle
=2^{-m}\sum_{y=0}^{2^m-1}Q_1(y)Q_2(y)$, it follows that
\[
\langle F,P\rangle
\ge
2^{-m}\left(\frac{2^m-1}{3}-1\right)
=
\frac13-\frac43\,2^{-m}.
\]

Meanwhile, Walsh orthogonality gives
$\langle F,P\rangle
=\sum_{|S|\le {\rm deg}(P)}\widehat F(S)\widehat P(S)$, while Parseval's identity
and $|P(y)|\le1$ give
$\sum_S\widehat P(S)^2
=2^{-m}\sum_yP(y)^2\le1$.
Cauchy--Schwarz therefore implies
\[
\frac13-o(1)
\le
\mathcal O\!\left(2^{-m/8}m^{3/4}\right)
\left(\sum_{j=0}^{{\rm deg}(P)}\binom mj\right)^{1/2},
\]
and hence
\[
\sum_{j=0}^{{\rm deg}(P)}\binom mj
=
\Omega\!\left(\frac{2^{m/4}}{m^{3/2}}\right).
\]

Let $H_2(t)=-t\log_2t-(1-t)\log_2(1-t)$ be the binary entropy
function. If ${\rm deg}(P)\le m/2$, the standard entropy bound
$\sum_{j=0}^{{\rm deg}(P)}\binom mj\le2^{mH_2({\rm deg}(P)/m)}$ yields, given the fact that
$H_2({\rm deg}(P)/m)\ge1/4-\mathcal O((\log m)/m)$. If ${\rm deg}(P)>m/2$, the following
conclusion is already automatic. Since $H_2$ is increasing on
$[0,1/2]$, we obtain
\begin{align}
    {\rm deg}(P)\ge
\bigl(H_2^{-1}(1/4)-o(1)\bigr)m,
\end{align}
where $H_2^{-1}$ denotes the inverse on $[0,1/2]$, $H_2^{-1}(1/4)\approx 0.041$ and $o(1)$
tends to zero as $p\to\infty$.

Finally, it remains to upper-bound $\deg(P)$ using the structure of $\CRAS$. Let $R$ be the number of adaptive rounds before the final
single-qubit readout.  Let $a$ be the auxiliary-qubit budget, with
the convention that every intermediate measurement available to the
classical controller is performed on these auxiliary qubits, each
auxiliary qubit produces at most one raw outcome per round. Thus, along every branch, the controller receives at most $Ra$ intermediate measurement bits,
counted before any classical processing or compression.  

Along a
fixed branch, the $Ra$ intermediate measurements and the final
output measurement provide at most $Ra+1$ single-qubit measurement
sites. Propagating any one of these sites backwards through a circuit
of two-qubit depth $d$ reaches at most $2^d$ input
qubits. The size of the involved branch therefore depends on at most
$(Ra+1)2^d$. Finally, $P$ is a signed sum of the branch
probabilities, and a sum of functions cannot have Fourier degree
larger than the maximum degree of its summands.  Hence
\[
    \deg(P)\leq (Ra+1)2^d .
\]

Combining the two bounds on $\deg(P)$ and using $m=n-1$, we find
\[
(Ra+1)2^d
\ge
\bigl(H_2^{-1}(1/4)-o(1)\bigr)(n-1),
\]
which proves the claimed lower bound.
\end{proof}

\section{Compilation to a two-dimensional architecture}
\label{app:2d-crossbar}

The arithmetic circuit admits an explicit row-column placement on a
two-dimensional nearest-neighbor lattice. Logical data qubits are arranged in
vertical columns, and each reversible Boolean update in a parallel arithmetic
layer is assigned a geometrically local block of ancilla qubits. Measurement-based fan-out
along the columns creates the required computational-basis copies of the source
qubits in that row. A clean one-dimensional nearest-neighbor dynamic subroutine
then evaluates the corresponding Boolean function along the row, and the
resulting bit is written back to the target qubit along its column. Hence the
compiled arithmetic circuit uses only nearest-neighbor coherent quantum
operations in two dimensions, with global classical feedforward as the only
nonlocal resource. The Fourier-transform blocks are handled by the
nearest-neighbor compilation theorem of Ref.~\cite{Rosenbaum2013}. Thus the
full construction has a two-dimensional geometrically local implementation,
while retaining only the polynomial-time classical controller allowed by the
adaptive model.

\begin{definition}[Two-dimensional geometric CRAS]
Let $\Lambda_n$ be a rectangular subgraph of the square lattice with
$|\Lambda_n|=\operatorname{poly}(n)$.  A polynomial-time uniform adaptive
circuit family belongs to $\mathrm{gCRAS}_2$ if all qubits occupy vertices of
$\Lambda_n$, every coherent two-qubit gate acts on a nearest-neighbour edge,
the total coherent quantum depth and the number of adaptive rounds are both
$\mathcal O(1)$, and intermediate single-qubit measurements, resets, and a global
polynomial-time classical controller are allowed.  
\end{definition}

\begin{lemma}[Two-dimensional crossbar for parallel controlled rotations]
\label{lem:rotation-crossbar}
Let $x_1,\ldots,x_m$ be control qubits and $q_1,\ldots,q_K$ be target qubits.
For polynomial-time computable angles $\theta_{ik}$, define
\begin{equation}
 U_{\Theta}
 :=\prod_{i=1}^{m}\prod_{k=1}^{K}
 \left(
 |0\rangle\!\langle0|_{x_i}\otimes I_{q_k}
 +|1\rangle\!\langle1|_{x_i}\otimes R_z(\theta_{ik})_{q_k}
 \right).
 \label{eq:parallel-controlled-rotations}
\end{equation}
There is a polynomial-time uniform two-dimensional nearest-neighbour dynamic
circuit that implements $U_{\Theta}$ exactly relative to the specified
controlled rotations as elementary two-qubit gates, uses $\mathcal O(mK)$
data-and-routing sites up to a constant factor, and has $\mathcal O(1)$ coherent depth
and $\mathcal O(1)$ adaptive rounds. 
\end{lemma}

\begin{figure*}[t]
  \centering
  \includegraphics[width=\textwidth]{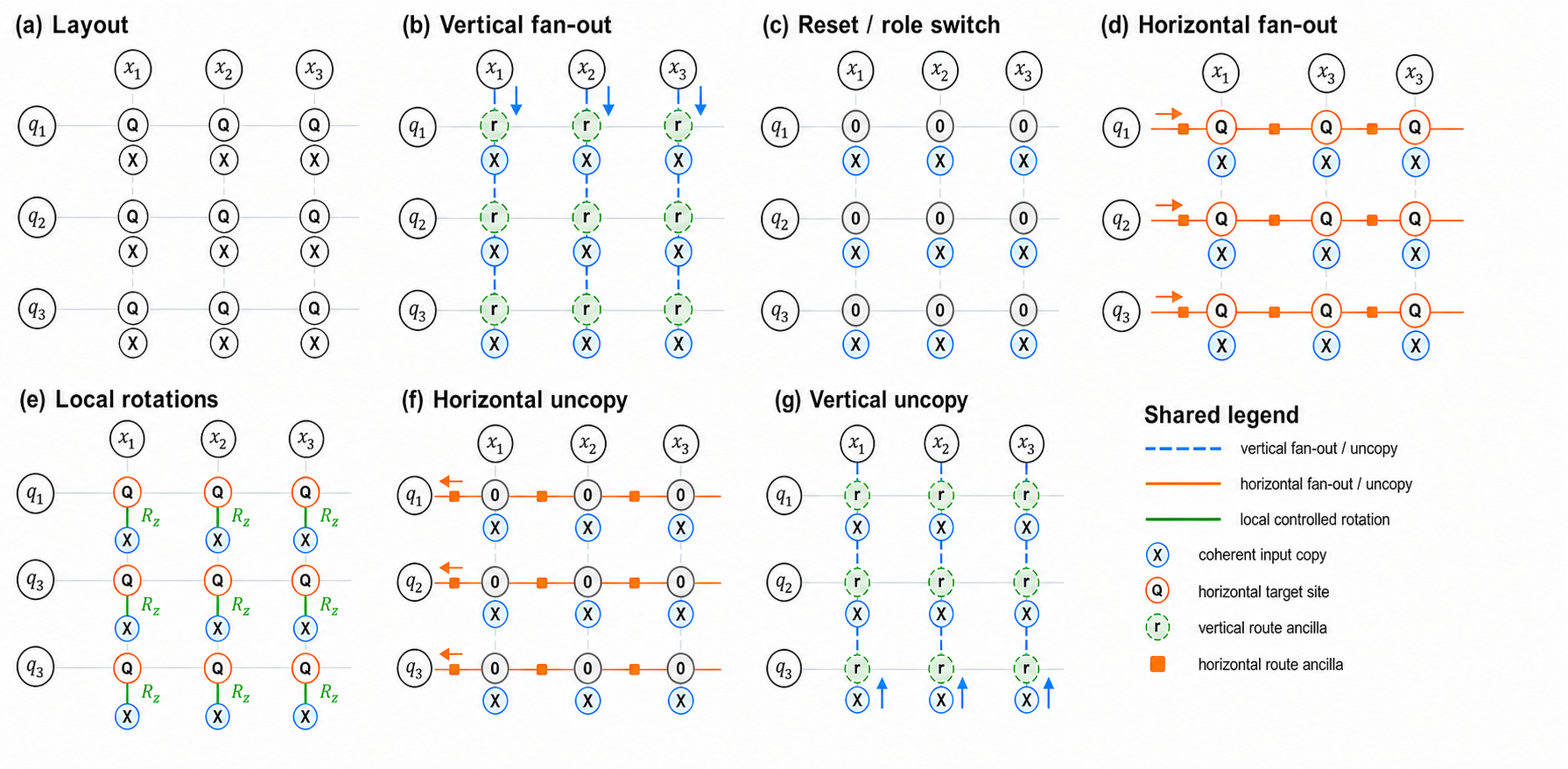}
  \caption{Seven-stage two-dimensional crossbar implementation of one
  parallel Hamming-weight-rotation layer.  (a) The crossbar contains adjacent
  sites $Q_{ik}$ and $X_{ik}$ at each cell.  (b) Vertical fan-out coherently
  copies $x_i$ to all $X_{ik}$ while the $Q_{ik}$ sites serve as route
  ancillas.  (c) Every $Q_{ik}$ is reset before changing roles.  (d)
  Horizontal fan-out expands each target $q_k$ into a cat state on
  $Q_{1k},\ldots,Q_{mk}$.  (e) The disjoint nearest-neighbour cell edges apply controlled
  $R_z(\theta_{ik})_{Q_{ik}}$ in parallel.  (f) Horizontal
  uncopy returns the $Q_{ik}$ sites to zero and leaves the accumulated rotation
  on $q_k$.  (g) Vertical uncopy erases the $X_{ik}$ copies.  Blue dashed rails
  represent vertical fan-out, orange solid rails represent horizontal
  fan-out, and green edges represent local controlled rotations.  The two rail
  families are never active simultaneously.  The small orange squares shown
  only in panels (d) and (f) are measured horizontal route ancillas and carry
  no logical data.}
  \label{fig:crossbar-seven-stages-static}
\end{figure*}

\begin{proof}
We use the one-dimensional dynamic fan-out construction of
Ref.~\cite[Sec.~II.B]{BaeumerWoerner2025}, in which system and route-ancilla qubits
alternate along a line.  It is an exact constant-depth adaptive implementation
of the multi-target CNOT; its route ancillas are measured, the prescribed
classical corrections are applied, and the route sites may be reset at the
end.  Since the multi-target CNOT is self-inverse, a second invocation with
clean route ancillas implements its inverse.

The crossing layout can be given explicitly.  Use lattice coordinates
\begin{equation}
 x_i=(2i,1),\qquad Q_{ik}=(2i,2k),\qquad
 X_{ik}=(2i,2k+1),\qquad q_k=(0,2k).
 \label{eq:crossbar-explicit-coordinates}
\end{equation}
For each fixed $i$, the vertical path is
$x_i,Q_{i1},X_{i1},Q_{i2},X_{i2},\ldots,Q_{iK},X_{iK}$.
Thus its system qubits are precisely $x_i,X_{i1},\ldots,X_{iK}$ and its
alternating route ancillas are precisely $Q_{i1},\ldots,Q_{iK}$.  For each
fixed $k$, the horizontal path starts at $q_k=(0,2k)$, uses the odd-coordinate
vertices $(2i-1,2k)$ as route ancillas, and uses the even-coordinate vertices
$Q_{ik}=(2i,2k)$ as its system targets.  Finally, $X_{ik}$ and $Q_{ik}$ are
nearest neighbours.  The two roles of $Q_{ik}$ occur in different stages: at
the end of a vertical fan-out every $Q_{ik}$ is measured and reset to
$|0\rangle$; after the inverse horizontal fan-out it is again $|0\rangle$ and
is available for the inverse vertical fan-out.  This time multiplexing, rather
than a simultaneous crossing of quantum paths, is the crossing gadget.

First, for every $i$, apply vertical fan-out from $x_i$ to
$X_{i1},\ldots,X_{iK}$.  All vertical paths are vertex-disjoint and run in
parallel.  Second, for every $k$, apply horizontal fan-out from $q_k$ to
$Q_{1k},\ldots,Q_{mk}$.  All horizontal paths are vertex-disjoint and run in
parallel.  Third, apply the nearest-neighbour gates
$\Lambda_{X_{ik}}R_z(\theta_{ik})_{Q_{ik}}$ at all cells simultaneously.
Fourth, invoke the horizontal multi-target CNOTs again to uncopy the $Q_{ik}$,
and fifth invoke the vertical multi-target CNOTs again to uncopy the $X_{ik}$.

To verify the logical action, fix a computational-basis value
$x=(x_1,\ldots,x_m)$.  Horizontal fan-out maps an arbitrary target state
$\alpha_k|0\rangle+\beta_k|1\rangle$ to a cat state.  The cell rotations
multiply its $|1\cdots1\rangle$ component by
$\exp(i\sum_i x_i\theta_{ik})$.  Horizontal unfan-out therefore leaves
$R_z(\sum_i x_i\theta_{ik})q_k$ and returns its target copies to zero.
Vertical uncopy then returns all $X_{ik}$ to zero.  This is exactly the action
of Eq.~\eqref{eq:parallel-controlled-rotations}; linearity gives the result for
arbitrary superpositions of the controls.  Each fan-out or unfan-out has
constant dynamic depth, and the local controlled rotations form one coherent
layer, proving the resource bounds. We summarize the construction step in Fig.~\ref{fig:crossbar-seven-stages-static}.
\end{proof}

\begin{lemma}[2D construction of Threshold gate]
\label{lem:geometric-HS-threshold}
Let $m=\operatorname{poly}(n)$, let $\tau\in\{0,\ldots,m\}$. There is a
polynomial-time uniform rectangular 2D dynamic
circuit implementing for
$\operatorname{TH}_{m,\tau}$. It non-trivially performs on ${\rm poly}(n)$ qubits, and its coherent depth and number of
adaptive rounds are $\mathcal O(1)$. The same statement holds for non-negative integer weighted thresholds
$\mathbf 1[\sum_i w_i x_i\ge\tau]$ whenever
$\sum_i w_i=\operatorname{poly}(n)$.
\end{lemma}

\begin{proof}
Set $S(x)=|x|$ and $W=m$ in the unweighted case, and set
$S(x)=\sum_iw_ix_i$ and $W=\sum_iw_i$ in the weighted case.
If $\tau>W$, the predicate is identically zero and the identity oracle proves
the claim, so assume $0\le\tau\le W$ below.
Expand the threshold gate using the construction of
Ref.~\cite{HoyerSpalek2005}.  For every
$t\in\{\tau,\tau+1,\ldots,W\}$, the construction of
$\operatorname{exact}[t]$ first prepares polynomially many qubits of the form
\begin{equation}
  H R_z\!\left(\phi_k\bigl(S(x)-t\bigr)\right)H|0\rangle,
  \label{eq:HS-first-rotation-layer}
\end{equation}
then applies a second Hamming-weight rotation to these qubits.  The first part
is a family of parallel rotations controlled by the original input bits, and
therefore has the form of Lemma~\ref{lem:rotation-crossbar}, with rows indexed
by the pair $(t,k)$; the term $-\phi_k t$ is a local unconditional rotation on
the corresponding target.  The second part is another crossbar layer whose controls
are the private intermediate qubits belonging to that value of $t$.
Different values of $t$ use disjoint workspaces and are evaluated in parallel.
The parity of
$\operatorname{exact}[\tau],\ldots,\operatorname{exact}[W]$ is computed by the
standard Hadamard conjugation of fan-out.  For a non-negative weighted
threshold, the first rotation layer replaces the cell angle $\phi_k$ by
$w_i\phi_k$; equivalently, one may first fan out $x_i$ into $w_i$ read-only
copies and invoke the unweighted construction on $W$ inputs.  Since
$W=\operatorname{poly}(n)$, both descriptions retain polynomial size and
constant logical depth.

It remains to compose these rotation layers into one fixed
two-dimensional threshold block.  Let $K_t$ denote the number of intermediate
qubits used by the selected $\operatorname{exact}[t]$ circuit.  Construct one
rectangular crossbar $R_1$ whose columns are indexed by the original input
qubits $x_i$ and whose rows are indexed by all pairs $(t,k)$, with
$t\in\{\tau,\ldots,W\}$ and $k\in\{1,\ldots,K_t\}$.  At the cell
$(i,t,k)$ apply the controlled rotation with angle $\phi_k$ in the unweighted
case and $w_i\phi_k$ in the weighted case; the shift
$R_z(-t\phi_k)$ is applied locally to the corresponding target.  Order the
rows lexicographically by $(t,k)$ and reflect the layout of
Lemma~\ref{lem:rotation-crossbar} so that the surviving intermediate qubits
$y_{t,k}$ lie on the right boundary of $R_1$.  The qubits belonging to each
fixed value of $t$ then form one consecutive boundary interval.  The vertical
and horizontal fan-out paths inside $R_1$ are activated in different dynamic
stages, exactly as in Lemma~\ref{lem:rotation-crossbar}, and hence no two
crossing quantum paths are simultaneously active.

For each $t$, attach to the corresponding boundary interval a rotated
crossbar $R_{2,t}$ whose controls are the qubits $y_{t,k}$ and whose targets
are the target qubits of the second Hamming-weight rotation in the
$\operatorname{exact}[t]$ circuit.  The rectangles $R_{2,t}$ occupy disjoint
horizontal slabs, so all values of $t$ are evaluated in parallel.  Any final
Hadamard gates, output negations, or constant-size bounded-fan-in
post-processing in the selected source circuit are placed inside the same
private slab.  Pad the slabs with unused lattice sites so that the resulting
output bit $e_t$ of every $\operatorname{exact}[t]$ circuit lies on a common
vertical boundary.  Thus the output of $R_1$ enters the second-layer
crossbars through adjacent boundary ports, and no source-to-block routing path
crosses the interior of another block.

Reserve an alternating system--route-ancilla strip along this common output
boundary.  Place $e_\tau,\ldots,e_W$ at designated system sites, place an
answer qubit $a$ at one endpoint, and initialize every unused system site on
the strip to $|0\rangle$.  Hadamard conjugation of the exact
one-dimensional dynamic fan-out gadget implements
\begin{equation}
  |e_\tau,\ldots,e_W\rangle|a\rangle
  \longmapsto
  |e_\tau,\ldots,e_W\rangle
  \left|a\oplus\bigoplus_{t=\tau}^{W}e_t\right\rangle .
  \label{eq:geometric-threshold-parity}
\end{equation}
The additional zero system sites are unchanged by this parity operation.
For the ideal exact predicates,
\begin{equation}
  \bigoplus_{t=\tau}^{W}\operatorname{exact}[t](x)
  =\mathbf 1[S(x)\ge\tau],
  \label{eq:exact-parity-is-threshold}
\end{equation}
because exactly one of these predicates is nonzero when $S(x)\ge\tau$, and
none is nonzero when $S(x)<\tau$.  Place the oracle target qubit $b$ adjacent
to $a$ and apply the nearest-neighbour gate
$\mathrm{CNOT}_{a\rightarrow b}$.  The parity strip, all rectangles
$R_{2,t}$, and finally $R_1$ are then run in reverse order, using freshly
reset route ancillas. 

At every dynamic stage the active paths are vertex-disjoint: they are either
the columns of $R_1$, the rows of $R_1$, paths contained in the mutually
disjoint rectangles $R_{2,t}$, or the single parity strip.  Boundary data
shared by consecutive blocks are never acted on by two gadgets
simultaneously.  Each crossbar and the parity strip has constant dynamic
depth, all $R_{2,t}$ run in parallel, and reversing the schedule changes the
depth by only a constant factor.  If $L_t$ denotes the number of targets in
the second Hamming-weight-rotation layer for value $t$, the occupied area is
bounded, up to a constant factor, by
\begin{equation}
  m\sum_{t=\tau}^{W}K_t
  +\sum_{t=\tau}^{W}K_tL_t
  +\sum_{t=\tau}^{W}(K_t+L_t).
  \label{eq:geometric-threshold-area}
\end{equation}
This quantity is polynomial in $n$, and padding the construction to its
rectangular bounding box remains polynomial.  The ordering of the boundary
ports, all lattice coordinates, and all rotation angles are computable in
polynomial time.  Hence the complete threshold block has a polynomial-time
uniform two-dimensional nearest-neighbour realization with polynomial area,
$\mathcal O(1)$ coherent depth, and $\mathcal O(1)$ adaptive rounds.
\end{proof}

\begin{lemma}[Explicit 2D compilation of arithmetic layers]
\label{lem:2d-clean-row-placement}
Consider a polynomial-time uniform, polynomial-size family of reversible
arithmetic circuits
\[
  \mathcal C
  =
  \mathcal C^{(L)}\cdots\mathcal C^{(2)}\mathcal C^{(1)}
\]
with $L=\mathcal O(1)$.
By introducing a fresh target qubit for each intermediate Boolean value and
using the standard compute--copy--uncompute construction, we may assume that
every layer $\mathcal C^{(\ell)}$ consists of a family of updates indexed by
a set $\mathcal A_\ell$:
\begin{equation}
  T_\alpha
  \longmapsto
  T_\alpha\oplus
  f_\alpha(X_{\alpha,1},\ldots,X_{\alpha,m_\alpha}),
  \qquad \alpha\in\mathcal A_\ell ,
  \label{eq:normalized-arithmetic-update}
\end{equation}
where the targets $T_\alpha$ are pairwise distinct and no target in the layer
is used as a source in the same layer:
\begin{equation}
 \{T_\alpha:\alpha\in\mathcal A_\ell\}
 \cap
 \{X_{\beta,j}:\beta\in\mathcal A_\ell,\,
                    1\le j\le m_\beta\}
 =\varnothing .
 \label{eq:source-target-separation}
\end{equation}
Assume that each $f_\alpha$ is either a weighted threshold function with
non-negative integer weights or the parity function. Then $\mathcal C$ has a
polynomial-time uniform 2D nearest-neighbour dynamic
implementation with polynomial area, $\mathcal O(1)$ coherent quantum depth,
and $\mathcal O(1)$ adaptive rounds. 
\end{lemma}

\begin{proof}
It is enough to consider one layer $\mathcal C^{(\ell)}$: write
$\mathcal A=\mathcal A_\ell$. Place the logical qubits along the upper
boundary of the grid, with one vertical routing path starting at each qubit.
For every $\alpha\in\mathcal A$, allocate a disjoint rectangular region. At
the intersection of this region with the path of $X_{\alpha,j}$, place a copy
qubit $C_{\alpha,j}$. The region also contains a clean output qubit
$F_\alpha$ and a reversible subcircuit that computes $f_\alpha$ into
$F_\alpha$. Threshold functions use
Lemma~\ref{lem:geometric-HS-threshold}, and parity is implemented by Hadamard
conjugation of exact dynamic fan-out, with the one-input case reducing to a
CNOT.

Each selected subcircuit is used in compute--copy--uncompute form,
\[
  U_\alpha
  =
  V_\alpha^\dagger
  \mathrm{CNOT}_{a_\alpha\to F_\alpha}
  V_\alpha,
  \qquad
  U_\alpha^2=I.
\]
The equality $U_\alpha^2=I$ holds for the selected subcircuit itself, whether
or not it approximates the ideal threshold function.

The layer is implemented in four steps:
\begin{enumerate}
  \item Apply exact dynamic fan-out along every vertical path, providing a
        local copy to each region that uses the corresponding source qubit.
  \item Apply all $U_\alpha$ in parallel, writing the output of each selected
        subcircuit into $F_\alpha$.
  \item Route each $F_\alpha$ horizontally to the vertical path of
        $T_\alpha$ and XOR it into $T_\alpha$.
  \item Apply all $U_\alpha$ again and then undo the source fan-outs. Since
        both operations are self-inverse, all work and copy qubits return to
        $|0\rangle$.
\end{enumerate}

For the third step, place a clean relay qubit $G_\alpha$ where the horizontal
and vertical paths meet. The sequence
\[
  \mathrm{CNOT}_{F_\alpha\to G_\alpha},\qquad
  \mathrm{CNOT}_{G_\alpha\to T_\alpha},\qquad
  \mathrm{CNOT}_{F_\alpha\to G_\alpha}
\]
implements $T_\alpha\mapsto T_\alpha\oplus F_\alpha$ and returns
$G_\alpha$ to $|0\rangle$. These long-range CNOTs are realized along
one-dimensional nearest-neighbour paths by the dynamic construction of
Lemma~\ref{fact:BW}, as used in Lemma~\ref{lem:rotation-crossbar}. The
horizontal paths lie in disjoint regions, the target paths are distinct, and
the source--target separation assumption ensures that a target path is not
used simultaneously for source distribution. Horizontal and vertical
routing are performed at different times, so the time-multiplexed crossing
layout of Lemma~\ref{lem:rotation-crossbar} removes the remaining geometric
conflicts.

For ideal function subcircuits, the map is
\begin{equation}
 |{\bf x}\rangle
 \bigotimes_{\alpha\in\mathcal A}|t_\alpha\rangle
 |0\rangle_{\mathrm{work}}
 \longmapsto
 |{\bf x}\rangle
 \bigotimes_{\alpha\in\mathcal A}
 |t_\alpha\oplus f_\alpha({\bf x}_\alpha)\rangle
 |0\rangle_{\mathrm{work}} .
 \label{eq:compiled-arithmetic-layer-action}
\end{equation}
If a selected approximate threshold subcircuit is used, the same schedule realizes
exactly the corresponding selected source channel. The fan-out and routing
constructions are exact and therefore introduce no additional approximation
error.

The four steps above contain only a constant number of adaptive rounds. Each
fan-out, function subcircuit, long-range CNOT, and uncomputation has
$\mathcal O(1)$ coherent depth and $\mathcal O(1)$ adaptive rounds, and all
instances of the same operation run in parallel on disjoint paths or regions.
There are polynomially many regions, each of polynomial area adding the
routing paths and constant-size crossing layouts, therefore leaves the total
area polynomial. The entire layout is computable in polynomial time from the
uniform description of $\mathcal C^{(\ell)}$. Finally,
$L=\mathcal O(1)$, as a result, composing the layers preserves all the claimed bounds.
\end{proof}

Finally, by Ref.~\cite{HoyerSpalek2005}, $\rm QFT$ can be compiled by a polynomial-size constant-depth implementation in the unbounded-fan-out model. This gives a constant-depth quantum circuit implementation after replacing fan-out by the CNOT gates with mid-circuit measurement~\cite{BaeumerWoerner2025}. Combined with the compilation method given by Ref~\cite{Rosenbaum2013}, we complete the proof of Corollary~\ref{coro:2Darchitecture}.

\vspace{10px}
\paragraph*{Proof of the entanglement area law result:}
It is useful to compare our setting with the area-law obstruction for
LOCC-assisted finite-depth circuits studied in Ref.~\cite{piroli2021quantum}.
In that work, the quantum part of the circuit is geometrically local: qubits (including physical qubits and ancilla qubits)
are placed on a fixed-dimensional lattice, the entangling gates are
nearest-neighbor gates of bounded local dimension, and the total quantum
depth is constant.  The circuit is further supplemented by local measurements,
arbitrary classical communication, and feedforward local unitaries.

The proof of the entanglement area-law is straightforward. Consider a bipartition $X=A\cup A^c$ of the lattice, and let $S_0(A)$ denote the Renyi-0 entropy across
this cut, equivalently the logarithm of the Schmidt rank. The initial state is
a product state, so $S_0(A)=0$.  In one geometrically local quantum layer,
only gates crossing the boundary $\partial A$ can increase the Schmidt rank
across the cut.  Since each such gate acts on a bounded-dimensional local
Hilbert space, it can increase $S_0(A)$ by at most a constant.  The number of
nearest-neighbor gates crossing the cut is $\mathcal O(|\partial A|)$.  Therefore one
quantum layer can increase $S_0(A)$ by at most $C|\partial A|$, for a
constant $C$ depending only on the local dimension and gate range.  Repeating
this argument for a constant number of quantum layers gives $S_0(A) \le C' |\partial A|$. The intermediate measurement, classical communication, and feedforward steps do
not invalidate this bound, because they are LOCC operations with respect to the
bipartition, and hence cannot increase bipartite entanglement.
Thus every state prepared by such a constant-depth geometrically adaptive quantum circuit obeys an entanglement area law.

This observation should be distinguished from the hardness phenomenon studied
in the present work.  The QCcc area-law result constrains the entanglement
structure of states prepared by geometrically local LOCC-assisted finite-depth
circuits.  Our result concerns the classical hardness of predicting a fixed
local observable in an adaptive shallow process.  In particular, the presence
of an area law for the prepared state does not by itself imply that all local
readout statistics generated by adaptive shallow dynamics are classically easy
to predict.

\end{document}